\documentclass[11pt,a4paper,english]{article}

\newcommand{\beq}{\begin{equation}}
\newcommand{\eeq}{\end{equation}}

\usepackage[latin1]{inputenc}
\usepackage{babel}
\usepackage[centertags]{amsmath}
\usepackage{amsfonts}
\usepackage{amssymb}
\usepackage{color}
\usepackage{amsthm}
\usepackage{newlfont}
\usepackage{graphicx}
\usepackage{dsfont}
\usepackage{geometry}
\usepackage{mathrsfs}
\usepackage{hyperref}

\DeclareMathOperator*{\esssup}{ess\,sup}

\newtheorem{theorem}{Theorem}[section]
\newtheorem{lemma}[theorem]{Lemma}
\newtheorem{coroll}[theorem]{Corollary}
\newtheorem{prop}[theorem]{Proposition}
\newtheorem{definition}[theorem]{Definition}
\newtheorem{remark}[theorem]{Remark}
\newtheorem{ass}[theorem]{Assumption}
\newcommand*{\lnr}{\left|\mkern-1mu\left|\mkern-1mu\left|}
\newcommand*{\rnr}{\right|\mkern-1mu\right|\mkern-1mu\right|}
\newcommand{\msc}[1]{\textbf{MSC2010 Classification:} #1.}
\newcommand{\keywords}[1]{\textbf{Key words:} #1.}

\makeatletter  
\@addtoreset{equation}{section}
\def\theequation{\arabic{section}.\arabic{equation}}
\makeatother  
\begin{document}
\title{\textbf{Analytical Pricing of American Put Options \\ on a Zero Coupon Bond \\ in the Heath-Jarrow-Morton Model}\thanks{These results extend a portion of the second Author PhD dissertation \cite{PhD-T} under the supervision of the first Author. In an earlier version of this work, circulated under the title ``Analytical Pricing of American Bond Options in the Heath-Jarrow-Morton Model'', we considered the positive part of the spot rate so to deal with a bounded stochastic discount factor. We acknowledge an anonymous referee for pointing out a possible inconsistency of this choice with the martingale theory of arbitrage-free pricing of derivatives. Here we consider the problem in its full generality.}}
\author{Maria B. Chiarolla\thanks{Dipartimento di Metodi e Modelli per l'Economia, il Territorio e la Finanza (MEMOTEF) Universit\`a di Roma `La Sapienza', via del Castro Laurenziano 9, 00161 Roma, Italy; \texttt{maria.chiarolla@uniroma1.it}}\:\:\:\:\:\:and\:\:\:\:\:Tiziano De Angelis\thanks{Corresponding author.~School of Mathematics, University of Manchester, Oxford Rd.~M13 9PL Manchester (UK); \texttt{tiziano.deangelis@manchester.ac.uk}}}
\maketitle
\begin{abstract}
We study the optimal stopping problem of pricing an American Put option on a Zero Coupon Bond (ZCB) in the Musiela's parametrization of the Heath-Jarrow-Morton (HJM) model for forward interest rates.

First we show regularity properties of the price function by probabilistic methods. Then we find an infinite dimensional variational formulation of the pricing problem by approximating the original optimal stopping problem by finite dimensional ones, after a suitable smoothing of the payoff. As expected, the first time the price of the American bond option equals the payoff is shown to be optimal.
\end{abstract}
\msc{91G80, 91G30, 60G40, 49J40, 35R15}
\vspace{+8pt}

\noindent\keywords{American Put options on a Bond, HJM model, forward interest rates, Musiela's parametrization, optimal stopping, infinite-dimensional stochastic analysis}

\section{Introduction}

A major challenge in mathematical finance is pricing derivatives with an increasing degree of complexity. A huge theoretical effort has been made in the last forty years to provide suitable tools for this purpose. The volume of traded options and the wide variety of their structures require a deep analysis of both theoretical and numerical methods. 

An important class of traded options is that of \textit{American options}. The mathematical formulation of this problem was given in the eighties by A.\ Bensoussan \cite{Ben84} and I.\ Karatzas \cite{Kar88}, among others. In mathematical terms pricing an American option corresponds to solving an \textit{optimal stopping} problem (for a survey cf.\ \cite{Pes-Shir}) in which the state dynamics is that of the security underlying the contract, usually a diffusion process (cf.~\cite{Pes-Shir}, Section 25, for a 1-dimensional geometric Brownian motion and \cite{Jai-Lam90} for more general diffusions). In such case one may find a variational formulation of the optimal stopping problem; that is, a free-boundary problem in the language of PDE (cf.~for instance \cite{Ben-Lio82} and \cite{Fri82} for a survey).

Here we aim to study the problem of pricing an American Put option on a Zero Coupon Bond (American Bond option) with the forward interest rate process as underlying. This option gives the holder the right to sell the ZCB for a fixed price $K$ at any time prior to the maturity $T$. The forward rate is the instantaneous interest rate agreed at time $t$ for a loan which will take place at a future time $s\geq t$. It is often denoted by $f(t,s)$ and taking $s=t$ one recovers the ``so called'' spot rate $R(t)=f(t,t)$. The price of the Bond, $B(t,s)$, is linked to the forward rate by the ordinary differential equation
\begin{eqnarray}\label{eq1}
f(t,s)=-\frac{\partial}{\partial s}\,\ln\left(B(t,s)\right).
\end{eqnarray}
For simplicity we will consider a ZCB with maturity equal to the maturity of the option, i.e.~$B(t,T)$. The option payoff at time $t$ is given by $[K-B(t,T)]^+$, where $[\,\cdot\,]^+$ denotes the positive part. The arbitrage free price of the American bond option is defined as
\begin{align}\label{eq-p}
V(t,f(t,\,\cdot\,))=\sup_{t\leq \tau\leq T}\mathbb{E}\Big\{e^{-\int_t^\tau{R(u)du}}\big[K-B(\tau,T)\big]^+\Big\}.
\end{align}
Notice that $V$ depends on the entire forward curve as it is typical of infinite dimensional optimization problems; hence one expects that it should solve an infinite dimensional variational inequality. However, for American options with an infinite dimensional underlying process it is not straightforward to establish a connection with PDE's in Hilbert spaces (cf.\ for instance \cite{DaPr-Zab04}). Such connection is instead known for \textit{European options} under forward rates; in fact their prices may be uniquely characterized through specific Kolmogorov equations (cf.\ \cite{Gol-Mus}). In some sense, that is a natural generalization of the Black and Scho\-les pricing formula to the infinite dimensional setting.  Infinite dimensional variational inequalities have not received as much attention as their finite dimensional counterparts. A good survey may be found in \cite{Bar-Mar08}, \cite{Ch-DeA12a}, \cite{Ch-Men89}, \cite{Gat-Sw99}, \cite{Zab01} and the references therein.

There exists a large literature on interest rate models concerning both theoretical and numerical aspects (for good surveys cf.\ \cite{Bjork}, \cite{Br-Mer}, \cite{Mu-Rut} for instance). In this paper, for the forward interest rates we choose the framework of the famous HJM model, one of the most reliable ones, which was introduced by D.~Heath, R.~Jarrow and A.~Morton \cite{HJM92} in 1992. The peculiarity of the stochastic process representing the forward interest rate is its infinite dimensional character. In essence, at each time $t$, the HJM model describes the family of rates $f(t,s)$, with $s\ge t$, that is the whole term structure of forward rates. A suitable parametrization of $f(t,s)$, modeled by an infinite dimensional stochastic differential equation, was obtained by M.\ Musiela \cite{Mus93} in 1993. An exhaustive description of the HJM model and its offspring may be found in \cite{Fil01} and \cite{Fil09}.

{In the last decade a significative effort has been made in order to establish conditions under which the forward rate curve of the HJM model admits a so-called finite-dimensional realization. In that case the forward curve may be described as a function of a finite dimensional diffusion (see~for instance \cite{Bj-Ch99}, \cite{Bj-La02}, \cite{Bj-La-Sv04}, \cite{ChKw03}, \cite{Fil-Te03}, \cite{Fil-Te03b}, \cite{Ta10}), and pricing American options reduces to solving variational inequalities in $\mathbb{R}^n$ in the spirit of \cite{Jai-Lam90}. Our problem (instead) is fully infinite dimensional. We do rely on a Galerkin-type finite-dimensional approximation of the forward curve but such reduction has no evident connection with the aforementioned theory.}

Our financial problem has been studied in \cite{Gat-Sw99} by means of viscosity theory, although in a different framework; that is, under the Goldys-Musiela-Sondermann parametrization (\cite{GMS}) of the HJM model. That completely determines the volatility structure of the dynamics and simplifies the underlying infinite dimensional stochastic differential equation by removing an unbounded term in the drift. A possible drawback of the model in \cite{Gat-Sw99} is the lack of consistency with the market's observations. This fact has been extensively discussed by D.\ Filipovic in \cite{Fil01}.

We provide a variational formulation of the pricing problem \eqref{eq-p} which is the infinite dimensional extension of that in \cite{Jai-Lam90}. 
We also find an optimal exercise time for the American Bond option. Our approach is partially based on our recent results on infinite dimensional optimal stopping and variational inequalities \cite{Ch-DeA12a}. However, here the payoff is less regular than the one studied in \cite{Ch-DeA12a} and the discount factor is stochastic, whereas in \cite{Ch-DeA12a} it was zero. To deal with the present setting we need to prove \emph{a-priori} regularity of $V$ in \eqref{eq-p} rather than obtaining it afterwards from the variational problem.

The paper is organized as follows. In Section 2 we introduce the financial model of the forward interest rate dynamics. In Section 3 we give the mathematical formulation of the corresponding optimal stopping problem and we carry out a detailed probabilistic analysis of the regularity properties of the American Bond option's price $V$. Section 4 is devoted to a regularization of the Put payoff $\Psi$. We associate an optimal stopping problem with value function $V_k$ to each smooth approximation $\Psi_k$ of the original payoff $\Psi$. Then we show that $V_k\to V$ as $k\to\infty$. In Section 5 we approximate the infinite dimensional optimal stopping problem $V_k$ by a sequence of finite dimensional ones. By using arguments as in \cite{Ch-DeA12a} we prove that $V_k$ is a suitable solution of an infinite dimensional variational inequality. Finally an infinite dimensional variational inequality for the price $V$ of the original American Bond option is obtained in Section 6. Also, we show that the first time at which $V$ equals the payoff $\Psi$ is an optimal exercise time for the option's holder. A technical appendix completes the paper.
\section{The interest rate model}

The forward rate at time $t$ for a loan taking place at a future time $s\geq t$ and returned at $s+ds$ is commonly denoted by $f(t,s)$. The instantaneous spot rate is obtained by setting $s=t$ and it is denoted by $R(t):=f(t,t)$. For every fixed maturity $s$ the time evolution of the forward rate is described by the map $t\mapsto f(t,s)$ with $t\leq s$. 

Consider a probability space $(\Omega,\mathcal{F},\mathbb{P})$ and denote by $(\mathcal{F}_t)_{t\ge0}$ the filtration generated by a $d$-dimensional Brownian motion
$B$, completed with the null sets. For simplicity but with no loss of generality we take $d=1$. Let $C^{0,1}_b(\mathbb{R})$ denote the set of bounded, Lipschitz-continuous real functions. Take $\tilde{\sigma}\in C^{0,1}_b(\mathbb{R})$, $\tilde{\sigma}$ non-negative and time-homogeneous {(other volatility structures which are possibly unbounded and non time-homogeneous are considered for instance in the recent papers \cite{BaZa12} and \cite{BaZa12b})}. According to the Heath-Jarrow-Morton model (HJM) (cf.~\cite{HJM92}), $\mathbb{P}$ may be assumed to be the risk-neutral probability measure on $(\Omega,\mathcal{F})$ and the forward rate with maturity $s$ may be described by the SDE
\begin{align}\label{SDE-HJM}\hspace{-5pt}
f(t,s)=f(0,s)&+\int_0^t\tilde{\sigma}(f(u,s))\int_u^s{\tilde{\sigma}(f(u,v))dv\,}du+\int_0^t{\tilde{\sigma}(f(u,s))
dB_u},\qquad t\in[0,s],
\end{align}
where $f(0,s)$ is deterministic and denotes the initial data at time zero. {The existence of a risk neutral probability measure is equivalent to assuming the particular form of the drift as given in \eqref{SDE-HJM} (cf.~\cite{HJM92} or \cite{Fil09}, Chapter 6).}

There exists a unique strong solution $f(\cdot,\cdot)$ of \label{SDE-HJM} continuous in both variables (cf.~\cite{Mor88}). Unfortunately the process $\big(f(t,s)\big)_{0\leq t\leq s}$ is not Markovian since the drift in \eqref{SDE-HJM} depends on the evolution of the whole forward curve. On the other hand, the Markov property holds for the infinite dimensional process $t\mapsto\{f(t,v),v\ge t\}$; therefore, pricing derivatives often requires to set dynamics in the infinite dimensional SDE's framework (cf.~\cite{DaPr-Zab}). This is accomplished by means of the so-called Musiela's parametrization (cf.~\cite{Mus93}) that describes the forward rate curve $f(t,s)$ in terms of the time to maturity $x:=s-t$ rather than the maturity time $s$; hence $f(t,s)=f(t,t+x)$. Then, in terms of the original forward curve we define the map $(t,x)\mapsto r_t(x)$ by setting $r_t(x):=f(t,t+x)$; that is, at any given time $t$ the model's input is the forward rate curve $x\mapsto r_t(x)$. The spot rate is obtained by taking $x=0$ and it is denoted by $r_t(0)$.

The process $t\mapsto r_t(\cdot)$ may be interpreted as an infinite-dimensional process taking values in a suitable Hilbert space $\mathcal{H}$. On such space the unbounded linear operator $A:={\frac{\partial}{\partial x}}$ generates a $C_0$-semigroup of bounded linear operators $\left\{S(t)\,|\,t\in\mathbb{R}_+\right\}$. In particular, $S(\cdot)$ is the semigroup of left-shifts defined by $S(t)h(x)=h(t+x)$ for any function $h:\mathbb{R}_+\rightarrow\mathbb{R}$ (for further details on semigroup theory the reader may refer to \cite{Pazy}). Define $\sigma(r_t)(x):=\tilde{\sigma}(f(t,s))$ and set
\begin{equation}\label{driftHJM}
F_\sigma(r_t)(x):=\sigma(r_t)(x)\int_0^x{\sigma(r_t)(y)dy},\qquad x\in\mathbb{R}_+.
\end{equation}
Straightforward calculations allow to write \eqref{SDE-HJM}
as
\begin{equation}\label{infiniteSDE}
r_t(x)=S(t)r_0(x)+\int_0^t{S(t-u)F_\sigma(r_u)(x)du}+\int_0^t{S(t-u)\sigma(r_u)(x)dB_u}.
\end{equation}
The link to the theory of infinite dimensional SDE's is now rather natural; in fact, under appropriate conditions on $\sigma$ and $\mathcal{H}$, (\ref{infiniteSDE}) turns out to be the unique \textit{mild} solution of
\begin{eqnarray}\label{infiniteSDEb}
\left\{
\begin{array}{l}
dr_t=\left[Ar_t+F_\sigma(r_t)\right]dt+\sigma(r_t)dB_t,\qquad t\in[0,T],\\
\\
r_0=r\in\mathcal{H},
\end{array}
\right.
\end{eqnarray}
where $0<T<\infty$ (cf.~\cite{DaPr-Zab}, Chapter 7).

In the present work $\mathcal{H}$ is chosen according to \cite{Fil01} {(cf.~also \cite{Fil-Te03}, Example 4.2)} and the notation $\mathcal{H}=\mathcal{H}_w$ is adopted (other possible models are available in \cite{BaZa12b}, \cite{Gol-Mus} and \cite{GMS}, among others). {We denote by $AC(\mathbb{R}_+)$ the set of absolutely continuous functions on $\mathbb{R}_+$}. Some fundamental facts are recalled in what follows.
\begin{definition}\label{hilb-sp}
Let $w:\mathbb{R}_+\rightarrow[1,+\infty)$ be a non decreasing $C^1$-function such that
\begin{equation}\label{def-w}
w^{-\frac{1}{3}}\in L^1(\mathbb{R}_+).
\end{equation}
Define
\begin{equation}\label{H-s}
\mathcal{H}_w:=\{h\in AC(\mathbb{R}_+)\,\,\,|\,\,\,\|h\|_w<\infty\},
\end{equation}
where
\begin{equation}\label{def-norm}
\|h\|^2_w:=|h(0)|^2+\int_{\mathbb{R}_+}{|h^\prime(x)|^2w(x)dx}.
\end{equation}
\end{definition}
\noindent The derivatives in Definition \ref{hilb-sp} are weak derivatives 
and the space $(\mathcal{H}_w,\|\cdot\|_w)$ is a Hilbert space (cf.~\cite{Fil01}, Theorem 5.1.1). An important consequence of \eqref{def-w} and \eqref{def-norm} is the  continuous injection  $\mathcal{H}_w\hookrightarrow L^\infty(\mathbb{R}_+)$ (cf.~\cite{Fil01}, Chapter 5, Eq.~(5.4)), i.e.~there exists $C>0$ such that
\begin{equation}\label{injection}
\sup_{x\in\mathbb{R}_+}|h(x)|\leq C\,\|h\|_w,\qquad h\in\mathcal{H}_w.
\end{equation}
Also, we point out that if $h\in\mathcal{H}_w$ then ${h(\infty)}:=\lim_{x\to\infty}h(x)$ exists and is finite (cf.~\cite{Fil01}, p.~77).

Recall that $\tilde{\sigma}\in C^{0,1}_b(\mathbb{R})$ and hence $\sigma$ inherits the same regularity as a function $\sigma:\mathbb{R}\to\mathbb{R}$. However in \eqref{infiniteSDEb} one must think of $\sigma$ and $F_\sigma$ as functions: $\mathcal{H}_w\to \mathcal{H}_w$. Define the set
\begin{equation}
\mathcal{H}_w^0:=\{h\in\mathcal{H}_w\,|\,h(\infty)=0\},
\end{equation}
then the following proposition holds (cf.~\cite{Fil01}, Chapter 5, Eq. (5.13)).
\begin{prop}\label{vol}
$\mathcal{H}_w^0$ is a closed subspace of $\mathcal{H}_w$. Moreover, $F_\sigma$ takes values in $\mathcal{H}_w$ if and only if $\sigma$ takes values in $\mathcal{H}_w^0$.
\end{prop}
\noindent As a consequence of Proposition \ref{vol} the only volatility structures allowed in \eqref{infiniteSDEb} are those such that $\sigma(h)(x)\rightarrow 0$ when $x\rightarrow\infty$ for any $h\in\mathcal{H}_w$. From now on we will make the following Assumption.
\begin{ass}\label{ass-sigma0}
The volatility $\sigma:\mathcal{H}_w\rightarrow\mathcal{H}^0_w$ is bounded and Lipschitz; i.e
\begin{align}\label{basic}
\|\sigma(h)\|_w<C_\sigma\qquad\textrm{and}\qquad\|\sigma(f)-\sigma(h)\|_w
\le L_\sigma\|f-h\|_w\quad\:\:\textrm{for all $f,h\in\mathcal{H}_w$}
\end{align}
and for some positive constants $C_\sigma$ and $L_\sigma$.
\end{ass}
A simple extension of \cite{Fil01}, Corollary 5.1.2, gives the following
\begin{prop}\label{lipschitz}
Under Assumption \ref{ass-sigma0} there exists $L_F>0$ depending on $C_\sigma$, $L_\sigma$ and such that
\begin{equation}\label{lipF}
\|F_\sigma(f)-F_\sigma(h)\|_w\leq L_F\|f-h\|_w,\qquad\textrm{for $f,h\in\mathcal{H}_w$.}
\end{equation}
\end{prop}
Now the main results of \cite{Fil01}, Chapter 5, may be summarized in the following
\begin{theorem}
Let $\mathcal{H}_w$ be as in \eqref{H-s}, then {the semigroup $\left\{S(t)\,|\,t\in\mathbb{R}_+\right\}$ is strongly continuous in $\mathcal{H}_w$ with infinitesimal generator denoted by $A$, where $D(A)=\{h\in\mathcal{H}_w\,|\,h^\prime\in\mathcal{H}_w\}$ and $Ah=h^\prime$.}
Moreover, under Assumption \ref{ass-sigma0} there exists a constant {$C_F>0$} such that
\begin{equation}\label{b-F}
\|F_\sigma(h)\|_{w}\leq C_F\,C_\sigma^2,\qquad\textrm{for all $h\in\mathcal{H}_w$.}
\end{equation}
\end{theorem}
The existence and uniqueness of the solution of (\ref{infiniteSDEb}) now follow.
\begin{theorem}\label{ExistInfinite}
Under Assumption \ref{ass-sigma0} there exists a unique mild solution of (\ref{infiniteSDEb}). 
\end{theorem}
\begin{proof}
The proof follows by standard arguments (cf.~\cite{DaPr-Zab}, Theorem 7.4) since $F_\sigma$ is bounded and Lipschitz by \eqref{b-F} and Proposition \ref{lipschitz}.
\end{proof}
\noindent The next Lemma provides standard estimates for the solution.
\begin{lemma}\label{lipSDE}
Let $r^h$ and $r^g$ be the mild solutions of (\ref{infiniteSDEb}) starting at $h$ and $g$, respectively. Then
\begin{eqnarray}
&&\mathbb{E}\Big\{\sup_{0\leq t\leq T}\|r^h_t\|^p_{w}\Big\}\leq C_{p,T}(1+\|h\|^p_{w}),\hspace{+20pt}\qquad1\leq p<\infty,\label{apr1}\\
&&\mathbb{E}\Big\{\sup_{0\leq t\leq T}\|r^h_t-r^g_t\|^p_{w}\Big\}\leq C_{p,T}\|h-g\|^p_{w},\hspace{+5pt}\qquad1\leq p<\infty,\label{lip}
\end{eqnarray}
where the positive constant $C_{p,T}$ depends only on $p$ and $T$.
\end{lemma}
\begin{proof}
The proof of \eqref{apr1} follows from \cite{DaPr-Zab}, Theorem 7.4, whereas the proof of \eqref{lip} is a consequence of \cite{DaPr-Zab}, Theorem 9.1 and a simple application of Jensen's inequality.
\end{proof}
\section{The pricing problem, some estimates, and regularity of the value function}\label{sec-pricing}

In terms of the unique solution of \eqref{infiniteSDEb}, the price at time $t$ of a Zero Coupon Bond (ZCB) with maturity $T\geq t$ may be expressed by
\begin{equation}\label{zcb}
B(t,{T};r_t(\cdot)):=\exp{\left(-\int_0^{{T}-t}{r_t(x)dx}\right)}.
\end{equation}
Recall that $r_\cdot(0)$ is the spot rate, then the stochastic discount factor $\Theta$ at time $t$ is
\begin{equation}\label{disc2}
\Theta(t;r_\cdot(0)):=\exp{\left(-\int_0^{t}{r_s(0)ds}\right)}.
\end{equation}

If the forward rate curve at time $t\in[0,T]$ is described by a function $h\in\mathcal{H}_w$, then the gain function at time $t$ of the American Put option with strike price $K<1$ and maturity $T$ is
\begin{eqnarray}\label{psi}
\Psi(t,h(\cdot)):=\left[K-B(t,T;h)\right]^+=\left[K-e^{-\int_0^{{T}-t}{h(x)dx}}\right]^+.
\end{eqnarray}
Let $r^{t,h}_s$, $s\geq t$ denote the value at time $s$ of the solution of \eqref{infiniteSDEb} with starting time $t$ and initial data $h$. The value function $V$ of the option evaluated at time $t\leq T$ may be written under the risk-neutral probability measure $\mathbb{P}$ as
\begin{equation}\label{priceHJM2}
V(t,h):=\sup_{t\leq\tau\leq T}\mathbb{E}\left\{e^{-\int_t^\tau{r^{t,h}_u(0)du}}\left[K-e^{-\int_0^{{T}-\tau}{r^{t,h}_\tau(x)dx}}\right]^+
\right\},
\end{equation}
where the supremum is taken over all stopping times with respect to the filtration $(\mathcal{F}^t_s)_{s\ge t}:=\sigma\big\{B_s-B_t\,;\,s\ge t\big\}$ generated by the increments of the Brownian motion driving \eqref{SDE-HJM}. Here the Markovian structure of the process $r^{t,h}$ implies that taking expectations conditioned to the Brownian filtration $\mathcal{F}_t$ at time $t$ is equivalent to unconditional expectations since $r^{t,h}_t=h$ is deterministic (cf.~for instance \cite{DaPr-Zab}, Chapter 9).

In what follows it will be sometimes convenient to write \eqref{priceHJM2} in terms of \eqref{disc2} and \eqref{psi} as
\begin{equation}\label{priceHJM3}
V(t,h)=\sup_{t\leq\tau\leq T}\mathbb{E}\left\{D(t,\tau;r^{t,h}_\cdot(0))\,\Psi(\tau,r^{t,h}_\tau(\cdot))
\right\},
\end{equation}
where $D(t,\tau;r_\cdot(0)):=\Theta(\tau;r_\cdot(0))/\Theta(t;r_\cdot(0))$.
Observe that 
the option pricing problem is meaningful only when the maturity of the option is lesser or equal than the maturity of the ZCB. In this work the two maturities are assumed to be equal for sake of simplicity and with no loss of generality.

Notice that both $\Psi$ and $V$ map $[0,T]\times\mathcal{H}_w$ into $\mathbb{R}$. Important regularity properties of $\Psi$ are described in the following
\begin{prop}\label{lip0}
The non-negative function $\Psi$ satisfies
\begin{equation}\label{psi1}
\sup_{(t,h)\in[0,T]\times\mathcal{H}_w}\Psi(t,h)\leq K<1.
\end{equation}
Moreover, there exist constants $C_1,C_2>0$ such that
\begin{eqnarray}
&&\left|\Psi(t,h)-\Psi(t,g)\right|\leq C_1\|h-g\|_w,\qquad\hspace{+13pt} g,h\in\mathcal{H}_w,\:t\in[0,T],\label{lip-psi1}\\
\nonumber\\
&&\left|\Psi(s,h)-\Psi(t,h)\right|\leq C_2\|h\|_w\,|t-s|,\qquad h\in\mathcal{H}_w,\:s,t\in[0,T],\:s\leq t.\label{lip-psi2}
\end{eqnarray}
\end{prop}
\begin{proof}
Define the function $\zeta(x):=[K-e^x]^+$ for $x\in\mathbb{R}$. It is not hard to check that $\|\zeta^\prime\|_{L^\infty(\mathbb{R})}\leq K<1$ with $\zeta^\prime$ the weak derivative of $\zeta$.
It follows that (cf.~for instance \cite{Brezis}, Chapter 8, Proposition 8.4)
\begin{equation}\label{lip1}
|\zeta(x)-\zeta(y)|\leq\|\zeta^\prime\|_{L^\infty(\mathbb{R})}\,|x-y|\leq|x-y|.
\end{equation}
Now define $X:=-\int_0^{{T}-t}{h(x)dx}$ and $Y:=-\int_0^{{T}-t}{g(x)dx}$, then \eqref{psi} and \eqref{lip1} give
\begin{align*}
\big|\Psi(t,h)-\Psi(t,g)\big|&=\big|\big[K-e^X\big]^+-\big[K-e^Y\big]^+\big|\leq |X-Y|\\
&\leq \int_0^{{T}-t}{|h(x)-g(x)|dx}\leq T\sup_{x\in\mathbb{R}_+}|h(x)-g(x)|\leq C\,{T}\,\|h-g\|_w,
\end{align*}
where the last inequality uses the continuous injection \eqref{injection}.

To prove \eqref{lip-psi2} take $s\leq t$ and proceed as above to obtain
\begin{eqnarray}
|\Psi(t,h)-\Psi(s,h)|\leq\int_{{T}-t}^{{T}-s}{|h(x)|dx}\leq\sup_{x\in\mathbb{R}_+}|h(x)|\,|t-s|\leq C\|h\|_w\,|t-s|.
\end{eqnarray}
\end{proof}

The following lemma provides a bound needed to obtain the regularity of the value function $V$. Such bound will be largely used in the paper.
\begin{lemma}\label{bound-discount}
There exist positive constants $\beta$ and $\gamma$ depending only on $C$ of \eqref{injection}, $C_\sigma$ of \eqref{basic}, $C_F$ of \eqref{b-F} and $T$ such that for every $0\le p \le 2$ and $h\in\mathcal{H}_w$ one has
\begin{align}\label{bound-r0}
\mathbb{E}\Big\{\sup_{t^\prime\le v\le T}e^{-p\int_{t^\prime}^v{r^{t^\prime,h}_t(0)dt}}\Big\}\le\gamma e^{\beta\|h\|_w}\qquad\text{for $t^\prime\in[0,T]$.}
\end{align}
\end{lemma}
\begin{proof}
For simplicity we prove \eqref{bound-r0} for $t^\prime=0$ and $p=1$ but the arguments of the proof apply to the general case due to time-homogeneity of \eqref{infiniteSDE} and the multiplicative character of $p$ in the exponential of the left-hand side of \eqref{bound-r0}. We fix $h\in\mathcal{H}_w$ and simplify notation by setting $r_t:=r^h_t$, $t\ge0$ and by denoting $B(h):=\exp\big\{C\,T\,\|h\|_w+C\,C_FC^2_\sigma T^2\big\}$. From \eqref{infiniteSDE} one easily obtains
\begin{align}\label{bound-dis0a}
\mathbb{E}\Big\{\sup_{0\le v\le T}e^{-\int_0^v{r_t(0)dt}}\Big\}\le&\: B(h)\mathbb{E}\Big\{\sup_{0\le v\le T}e^{-\int^v_0{\big(\int^t_0{S(t-s)\sigma(r_s)(0)dB_s}\big)}dt}\Big\}\nonumber\\
=&\:B(h)\mathbb{E}\Big\{\sup_{0\le v\le T}e^{-\int^T_0{\big(I_{\{t<v\}}\int^t_0{S(t-s)\sigma(r_s)(0)dB_s}\big)}dt}\Big\}.
\end{align}
Then \eqref{bound-dis0a}, convexity of the exponential function and Jensen's inequality give
\begin{align}\label{bound-dis0}
\mathbb{E}\Big\{\sup_{0\le v\le T}e^{-\int_0^v{r_t(0)dt}}\Big\}\le&\: B(h)\mathbb{E}\Big\{\sup_{0\le v\le T}\frac{1}{T}\int^T_0{e^{-T\,I_{\{t<v\}}\int^t_0{S(t-s)\sigma(r_s)(0)dB_s}}dt}\Big\}\nonumber\\
=&\: B(h)\mathbb{E}\Big\{\sup_{0\le v\le T}\frac{1}{T}\left(\int^v_0{e^{-T\,\int^t_0{S(t-s)\sigma(r_s)(0)dB_s}}dt}+\int^T_v\,1\,dt\right)\Big\}\nonumber\\
\le&\: B(h)\Big(\frac{1}{T}\int^T_0{\mathbb{E}\Big\{e^{-T\int^t_0{S(t-s)\sigma(r_s)(0)dB_s}}\Big\}dt}+1\Big).
\end{align}
In order to find an upper bound for the expectation in the last line of \eqref{bound-dis0}, for every $t\in[0,T]$ fixed we define the square-integrable, continuous martingale process
\begin{align*}
M^t_v:=\int^{v\wedge t}_0{S(t-s)\sigma(r_s)(0)dB_s},\qquad v\ge0;
\end{align*}
then, for arbitrary $R>0$, we define the stopping time
\begin{align*}
\tau^t_R:=\inf\big\{v\ge0\,:\,e^{-T\,M^t_v}> R\big\}\wedge T.
\end{align*}
Clearly $\tau^t_R\to T$ $\mathbb{P}$-a.s.~as $R\to\infty$. An application of It\^o's formula gives, for $v\in[0,T]$,
\begin{align}\label{bound-dis3}
e^{-T\,M^t_{v\wedge\tau^{t}_R}}=1-T\int^{v\wedge\tau^{t}_R}_0{e^{-T\,M^t_u}dM^t_u}+\frac{1}{2}T^2
\int^{v\wedge\tau^{t}_R}_0{e^{-T\,M^t_u}d\langle M^t\rangle_u}\,.
\end{align}
In this simple case the quadratic variation $\langle M^t\rangle$ of $M^t$ reads
\begin{align*}
\langle M^t\rangle_v=\int^{v\wedge t}_0{\big(S(t-s)\sigma(r_s)(0)\big)^2ds}.
\end{align*}
The stochastic integral in \eqref{bound-dis3} is a real martingale since the integrand is bounded by $R$, hence by taking expectations one finds
\begin{align*}
\mathbb{E}\Big\{e^{-T\,M^t_{v\wedge\tau^{t}_R}}\Big\}=&1+\frac{1}{2}T^2
\mathbb{E}\Big\{\int^{v\wedge\tau^{t}_R}_0{e^{-T\,M^t_u}d\langle M^t\rangle_u}\Big\}\le1+\frac{1}{2}T^2C^2C^2_\sigma
\mathbb{E}\Big\{\int^{v\wedge\tau^{t}_R}_0{e^{-T\,M^t_u}du}\Big\}
\end{align*}
with $C_\sigma$ as in \eqref{basic} and $C$ as in \eqref{injection}. Then the limit as $R\uparrow\infty$ gives
\begin{align*}
\mathbb{E}\Big\{e^{-T\,M^t_{v}}\Big\}=&\mathbb{E}\Big\{\liminf_{R\to\infty} e^{-T\,M^t_{v\wedge\tau^{t}_R}}\Big\}\le\liminf_{R\to\infty}\mathbb{E}\Big\{e^{-T\,M^t_{v\wedge\tau^{t}_R}}\Big\}
\\
\le&1+\frac{1}{2}T^2C^2C^2_\sigma
\lim_{R\to\infty}\mathbb{E}\Big\{\int^{v\wedge\tau^{t}_R}_0{e^{-T\,M^t_u}du}\Big\}=1+\frac{1}{2}T^2C^2C^2_\sigma
\mathbb{E}\Big\{\int^{v}_0{e^{-T\,M^t_u}du}\Big\}\nonumber
\end{align*}
by Fatou's lemma and monotone convergence theorem. Finally, an application of Gronwall's lemma provides
\begin{align}\label{bound-dis5}
\mathbb{E}\Big\{e^{-T\int^t_0{S(t-s)\sigma(r_s)(0)dB_s}}\Big\}=\mathbb{E}\Big\{e^{-T\,M^t_{t}}\Big\}\le e^{\frac{1}{2}T^3C^2C^2_\sigma},
\end{align}
when $v=t$. Hence \eqref{bound-dis0} and \eqref{bound-dis5} imply \eqref{bound-r0} for suitable $\gamma$ and $\beta$, which however may be taken to be independent of $p$ since $p$ is bounded.
\end{proof}

We now find some regularity properties of the value function $V$ by employing purely probabilistic arguments. The Lipschitz continuity with respect to time shown below is a remarkable and rather unusual feature in optimal stopping which follows from the peculiar structure of our problem. 
\begin{theorem}\label{reg-v-fun}
The value function $V$ is locally bounded; that is,
\begin{equation}\label{bdd-V}
0\le V(t,h)\leq K\,\gamma\, e^{\beta\|h\|_w}\qquad h\in\mathcal{H}_w,\,t\in[0,T],
\end{equation}
with $\gamma$ and $\beta$ as in Lemma \ref{bound-discount}. Moreover, there exist $L_V>0$ and $L'_V>0$ such that
\begin{equation}\label{lip-V}
|V(t,h)-V(t,g)|\leq L_V\Big(e^{\frac{\beta}{2}\|h\|_w}+e^{\frac{\beta}{2}\|g\|_w}\Big)\|h-g\|_{w},\qquad h,g\in\mathcal{H}_w,\,t\in[0,T],
\end{equation}
and
\begin{equation}\label{lip-Vt}
|V(t_2,h)-V(t_1,h)|\leq L'_V\, (1+\|h\|_w)e^{\frac{\beta}{2}\|h\|_w}\,(t_2-t_1),\qquad h\in\mathcal{H}_w,\,0\le t_1<t_2\le T.
\end{equation}
\end{theorem}
\begin{proof}
The proof is quite long and it is provided in Appendix \ref{app-1}.
\end{proof}
\begin{coroll}\label{continuity-V}
The map $(t,h)\mapsto V(t,h)$ is continuous on $[0,T]\times\mathcal{H}_w$.
\end{coroll}
It is worth noticing that in \cite{Ch-DeA12a} continuity of the value function was obtained instead by regularity results for solutions of variational inequalities. However, in that paper no Lipschitz continuity in time was established.

Before concluding this section we provide a uniform integrability result which will be useful to characterize existence of optimal stopping times in what follows. Fix $t\in[0,T]$, $h\in\mathcal{H}_w$, and define
\begin{align}\label{family}
Y^{t,h}_\tau:=e^{-\int^\tau_t{r^{t,h}_u(0)du}}\,V(\tau,r^{t,h}_\tau)
\end{align}
where $\tau\in[t,T]$ is an arbitrary stopping time.
\begin{prop}\label{unif-int}
For any $(t,h)\in[0,T]\times\mathcal{H}_w$, the family $\big\{Y^{t,h}_\tau\,:\,\tau\in[t,T]\:\text{stopping time}\big\}$ is uniformly integrable.
\end{prop}
\begin{proof}
The random variables $Y^{t,h}_\tau$ are clearly positive and it suffices to show that their $L^2(\Omega,\mathbb{P})$-norm is uniformly bounded. The random variable $V(\tau,r^{t,h}_\tau)$ may be expressed in terms of an essential supremum and
\begin{align}\label{unif-int01}
\mathbb{E}\Big\{\big|Y^{t,h}_\tau\big|^2\Big\}:=&\mathbb{E}\bigg\{\bigg|e^{-\int^\tau_t{r^{t,h}_u(0)du}}
\esssup_{\tau\le\sigma\le T}\mathbb{E}\bigg\{e^{-\int^\sigma_\tau{r^{t,h}_u(0)du}}\Psi(\sigma,r^{t,h}_\sigma)\Big|\mathcal{F}_\tau\bigg\}\bigg|^2
\bigg\}\nonumber\\
=&\mathbb{E}\bigg\{\bigg|\esssup_{\tau\le\sigma\le T}\mathbb{E}\bigg\{e^{-\int^\sigma_t{r^{t,h}_u(0)du}}\Psi(\sigma,r^{t,h}_\sigma)\Big|\mathcal{F}_\tau\bigg\}\bigg|^2
\bigg\}\\
\le& K^2\,\mathbb{E}\bigg\{\sup_{t\le v\le T}e^{-2\int^v_t{r^{t,h}_u(0)du}}\bigg\}\le \gamma\,K^2\,e^{\beta\|h\|_w}\nonumber
\end{align}
by Lemma \ref{bound-discount} for suitable $\gamma>0$ and $\beta>0$. 
\end{proof}
\section{Preliminary smoothing of the gain function}\label{sec-est}
The pricing problem \eqref{priceHJM2} is an optimal stopping problem involving a stochastic discount factor and a gain function $\Psi$ (cf.~\eqref{psi}) which is not smooth enough to allow a straightforward application of the results in \cite{Ch-DeA12a}. 
It is then natural to tackle problem \eqref{priceHJM2} by considering a regularized version of $\Psi$. For that we now introduce appropriate infinite-dimensional Sobolev spaces.
\subsection{Gaussian measure and Sobolev spaces}
Recall that since $A$ is the infinitesimal generator of a $C_0$-semigroup on $\mathcal{H}_w$ then it is a closed operator and its domain $D(A)$ is dense in $\mathcal{H}_w$ (cf.~for instance \cite{DaPr-Zab}, Appendix A.2). Hence one can pick an orthonormal basis of $\mathcal{H}_w$, denoted $\{\varphi_1,\varphi_2,\ldots\}$, with $\varphi_i\in D(A)$, $i=1,2,\ldots$. We now define a trace class operator which will play a crucial role.
\begin{definition}\label{Q-op}
Let $Q:\mathcal{H}_w\to\mathcal{H}_w$  be the positive, linear operator defined by
\begin{equation*}
Q\varphi_i=\lambda_i\varphi_i, \qquad\lambda_i>0,\qquad i=1,2,\ldots,
\end{equation*}
and such that $\sum^\infty_{i=1}{\lambda_i}<\infty$ (i.e. it is of trace class).
\end{definition}

Define the centered Gaussian measure $\mu$ with covariance operator $Q$ (cf. \cite{Bogachev}, \cite{DaPr}, \cite{DaPr-Zab04}); that is, the restriction to the vectors\footnote{$\ell_2$ denotes the set of infinite vectors $h:=(h_1,h_2,\ldots)$ such that $\sum_{k}{h_k^2}<+\infty$.} $h\in\ell_2$ of the infinite product measure
\begin{eqnarray}\label{gaussmeas}
\mu(dh)=\prod^\infty_{i=1}{\frac{1}{\sqrt{2\pi\lambda_i}}e^{-\frac{h^2_i}{2\lambda_i}}dh_i}.
\end{eqnarray}
For $1\leq p<+\infty$ and $f:\mathcal{H}_w\to\mathbb{R}$, define the $L^p(\mathcal{H}_w,\mu)$ norm as
\begin{align}
&\|f\|_{L^p({\mathcal{H}_w,\mu})}:=\left(\int_{\mathcal{H}_w}{|f(h)|^p\mu(dh)}\right)^{\frac{1}{p}}\qquad\text{for $1\le p<+\infty$.}\label{Lp-n1}
\end{align}
Then, with the notation of \cite{DaPr}, Chapter 10, we consider derivatives in the Friedrichs sense; that is,
\begin{align}\label{Dfried}
D_k\,f(h):=\lim_{\varepsilon\to 0}\frac{1}{\varepsilon}\left[f(h+\varepsilon\varphi_k)-f(h)\right],\qquad h\in\mathcal{H}_w,\,k\in\mathbb{N},
\end{align}
when the limit exists. For $Df(h):=\left(D_1f(h),\,D_2f(h),\,\ldots\,\right)$, we say that $Df(h)\in\mathcal{H}_w$ for given $h\in\mathcal{H}$ if
\begin{align}\label{Dfried02}
\big\|Df(h)\big\|_{w}:=\Big(\sum_{k}\big|D_kf(h)\big|^2\Big)^{\frac{1}{2}}<+\infty.
\end{align}
The $L^p(\mathcal{H}_w,\mu;\mathcal{H}_w)$ norm of $Df$ is defined as
\begin{align}
&\|Df\|_{L^p({\mathcal{H}_w,\mu};\mathcal{H})}:=\left(\int_{\mathcal{H}_w}{\|Df(h)\|^p_{w}\,\mu(dh)}\right)
^{\frac{1}{p}}\qquad\text{for $1\le p<+\infty$},\label{Lp-n2}
\end{align}
One can show that $D$ is closable in $L^p(\mathcal{H}_w,\mu)$ (cf.\ \cite{DaPr}, Chapter 10). Let $\overline{D}$ denote the closure of $D$ in $L^p(\mathcal{H}_w,\mu)$ and define the Sobolev space
\begin{eqnarray}\label{def-W1p}
W^{1,p}(\mathcal{H}_w,\mu):=\{f: f\in L^p(\mathcal{H}_w,\mu)\,\text{and}\,\overline{D}f\in L^p(\mathcal{H}_w,\mu;\mathcal{H}_w)\}.
\end{eqnarray}
Notice however that in the case of generalized derivatives $D$ and $\overline{D}$ are the same.

For $n\in\mathbb{N}$ the finite dimensional counterpart of $\mu$, $L^p({\mathcal{H}_w,\mu})$, $L^p(\mathcal{H}_w,\mu;\mathcal{H}_w)$ are, respectively,
\begin{equation*}
\mu_n(dh):=\prod^n_{i=1}{\frac{1}{\sqrt{2\pi\lambda_i}}e^{-\frac{h^2_i}{2\lambda_i}}dh_i},\quad L^p(\mathbb{R}^n,\mu_n),\quad L^p(\mathbb{R}^n,\mu_n;\mathbb{R}^n).
\end{equation*}
\begin{remark}\label{conn}
If $f:\mathbb{R}^n\to\mathbb{R}$, then
\begin{align*}
\|f\|_{L^p({\mathcal{H}_w,\mu})}=\left(
\int_{\mathbb{R}^n}{|f(h)|^p\mu_n(dh)}\right)^{\frac{1}{p}}=:\|f\|_{L^p({\mathbb{R}^n,\mu_n})}
\end{align*}
and
\begin{align*}
\|Df\|_{L^p({\mathcal{H}_w,\mu};\mathcal{H}_w)}=
\left(\int_{\mathbb{R}^n}{\|Df(h)\|^p_{\mathbb{R}^n}\,\mu_n(dh)}\right)^{\frac{1}{p}}=:
\|Df\|_{L^p({\mathbb{R}^n,\mu_n};\mathbb{R}^n)}\,.
\end{align*}
Notice as well that $\mu_n$ is equivalent to the Lebesgue measure on $\mathbb{R}^n$ for $n>0$.
\end{remark}
Again as in \cite{DaPr}, Chapter 10, we define
\begin{align}\label{Dfried03}
D_kD_jf(h):=\lim_{\varepsilon\to0}\frac{1}{\varepsilon}\left[D_jf(h+\varepsilon\varphi_k)-D_jf(h)\right],\qquad h\in\mathcal{H},\,k\in\mathbb{N},
\end{align}
when the limit exists in $\mathcal{H}_w$. As usual $D^2f:\mathcal{H}_w\to \mathcal{L}(\mathcal{H}_w)$ where $\mathcal{L}(\mathcal{H}_w)$ denotes the space of linear operators on $\mathcal{H}_w$. In this paper we do not need an $L^p$-space associated to the second derivative.

The next proposition provides useful bounds on the gain function $\Psi$ and its proof may be found in Appendix \ref{app0}.
\begin{prop}\label{reg-psi2}
For $1\le p<+\infty$ there exists a positive constant $C_{\Psi,p}$ such that
\begin{equation}\label{rp1-2}
\sup_{t\in[0,T]}\|\Psi(t)\|_{W^{1,p}({\mathcal{H}_w},\mu)}<C_{\Psi,p}\qquad\text{and}
\qquad\int_0^T{\bigg\|\frac{\partial\Psi}{\partial t}(t)\bigg\|^2_{L^p({\mathcal{H}_w},\mu)}dt}<C_{\Psi,p}\,.
\end{equation}
\end{prop}
It is now crucial to observe that
\begin{align*}
\int_{\mathcal{H}_w}{\|h\|^p_w\,\,e^{\lambda\|h\|_w}\mu(dh)}<\infty
\end{align*}
for arbitrary $\lambda\in\mathbb{R}$ and $1\le p<+\infty$ and therefore Theorem \ref{reg-v-fun} and arguments similar to those employed in the proof of Proposition \ref{reg-psi2} provide the following corollary for the value function $V$ of problem \eqref{priceHJM2}.
\begin{coroll}\label{arto}
For $1\le p<+\infty$ there exists a positive constant $C_{V,p}$ such that
\begin{equation}\label{rp1-V}
\sup_{t\in[0,T]}\|V(t)\|_{W^{1,p}({\mathcal{H}_w},\mu)}<C_{V,p}\qquad\text{and}\qquad\int_0^T{\bigg\|\frac{\partial V}{\partial t}(t)\bigg\|^2_{L^p({\mathcal{H}_w},\mu)}dt}<C_{V,p}\,.
\end{equation}
\end{coroll}
\subsection{Smoothing the gain function}

The smoothing procedure we introduce in this section will be obtained as a slight generalization of that used in \cite{Ma-Rock}, Chapter 4, Lemma 4.1. Define the family $(\Phi_t)_{t\in[0,{T}]}\subset{\mathcal{H}_w}^{\hspace{-5pt}*}$ by
\begin{equation}
\Phi_t(h):=-\int_0^{{T}-t}{h(x)dx},\qquad h\in{\mathcal{H}_w},\:t\in[0,T].
\end{equation}
By continuous injection \eqref{injection} follows
\begin{eqnarray*}
|\Phi_t(h)|\leq C_{{T}}\|h\|_{w},\qquad h\in{\mathcal{H}_w},
\end{eqnarray*}
for a positive constant $C_T$, i.e. $(\Phi_t)_{t\in[0,{T}]}\subset{\mathcal{H}_w}^{\hspace{-5pt}*}$ is a bounded subset of ${\mathcal{H}_w}^{\hspace{-5pt}*}$ with bound $C_{{T}}$.

Let $C^{1,2}_b([0,T]\times{\mathcal{H}_w})$ be the set of bounded continuous functions which are continuously differentiable once with respect to time and twice with respect to the space variable (in the Friedrichs sense) with bounded derivatives (cf.~\eqref{Dfried} and \eqref{Dfried03}).
\begin{prop}\label{mollify1}
There exists a sequence $(\Psi_k)_{k\in\mathbb{N}}\subset C^{1,2}_b([0,T]\times{\mathcal{H}_w})$ satisfying \eqref{lip-psi1} and \eqref{lip-psi2}, such that
\begin{equation}\label{uni-psi}
\sup_{(t,h)\in[0,{T}]\times{\mathcal{H}_w}}\left|\Psi_k(t,h)-\Psi(t,h)\right|\leq\frac{1}{k}\quad\text{for all $k\in\mathbb{N}$}
\end{equation}
and, for $1\le p <+\infty$,
\begin{align}
&\Psi_k\to\Psi\quad \textrm{as $k\to\infty$ in $L^p(0,T;L^p({\mathcal{H}_w},\mu))$},\label{conv-psi1}\\[+8pt]
&D\Psi_k\to D\Psi\quad \textrm{as $k\to\infty$ in $L^p(0,T;L^p({\mathcal{H}_w},\mu;{\mathcal{H}_w}))$},\label{conv-psi2}\\[+4pt]
&\frac{\partial}{\partial t}\Psi_k\to \frac{\partial}{\partial t}\Psi\quad \textrm{as $k\to\infty$ in $L^2(0,T;L^p({\mathcal{H}_w},\mu))$}.\label{conv-psi3}
\end{align}
\end{prop}
\begin{proof}
The gain function $\Psi$ in \eqref{psi} is the composition of $v:[0,T]\times{\mathcal{H}_w}\to\mathbb{R}$, with
\begin{eqnarray*}
v(t,h):=K-e^{\Phi_t(h)},
\end{eqnarray*}
and $g:\mathbb{R}\to\mathbb{R}_+$, with $g(z):=[z]^+$.

Define $I:=Im(v)=(-\infty,K)$ and notice that $g|_{I}:(-\infty,K)\to[0,K)$, where $g|_I$ is the restriction of $g$ to the domain $I$. Let $C^\infty_c(I)$ be the set of functions with compact support on $I$ and continuously differentiable infinitely many times. Take the standard mollifiers $(\rho_k)_{k\in\mathbb{N}}\subset C^\infty_c(I)$ (as in \cite{Brezis}, Chapter 4, p.~108) and consider the mollified sequence $(g_k)_{k\in\mathbb{N}}\subset C^\infty_c(I)$, where $g_k:=\rho_k\star g$. Since $g\in W^{1,p}(I)$ for all $1\leq p\le \infty$, then $g_k\to g$ in $W^{1,p}(I)$, $1\leq p<\infty$ as $k\to\infty$. It is well known that $g^\prime_k=\rho_k\star g^\prime$, where $g^\prime$ represents the weak derivative of $g$, moreover $g_k\to g$ and $g^\prime_k\to g^\prime$ pointwise as $k\to\infty$. The convergence is also locally uniform on every compact subset of $\mathbb{R}$, i.e. $\|g_k-g\|_{L^\infty(\,\bar{I}\,)}\to 0$, as $k\to\infty$ for any compact $\bar{I}\subset\mathbb{R}$. It is not hard to prove that
\begin{equation}\label{aste}
\textrm{$\|g_k\|_{L^\infty(I)}\leq\|g\|_{L^\infty(I)}=K$\:\:\: and\:\:\: $\|g^\prime_k\|_{L^\infty(I)}\leq\|g^\prime\|_{L^\infty(I)}=1$,}
\end{equation}
since $g$ and its weak derivative $g^\prime$ are both uniformly bounded on $I$.

Notice that $v\in C^{1,2}_b([0,{T}]\times{\mathcal{H}_w})$ and therefore\footnote{Differentiability with respect to time is obvious. Since ${\mathcal{H}_w}$ is identified to its dual then $\Phi_t\in\mathcal{H}_w$ with coordinates $(\Phi^1_t,\Phi^2_t,\ldots)$ and $\|\Phi_t\|^2_{w}=\sum_{i=1}^\infty{(\Phi_t^i)^2}$. Hence $\|Dv(t,h)\|^2_{w}=\sum_{i=1}^\infty{(D_iv(t,h))^2}=e^{2\Phi_t(h)}\sum_{i=1}^\infty{(\Phi_t^i)^2}
=e^{2\Phi_t(h)}\|\Phi_t\|_{w}$. Similarly $\langle D^2v(t,h)\,u,\bar{u}\rangle\leq e^{\Phi_t(h)}\|\Phi_t\|^2_{w}\|u\|_{w}\|\bar{u}\|_{w}$.}
\begin{align}
\frac{\partial (g_k\circ v)}{\partial t}(t,h)&=g^\prime_k(v)\,\frac{\partial v}{\partial t}(t,h)\:\:\:\:\text{and}\:\:\:\:
D(g_k\circ v)(t,h)=g^\prime_k(v)\,Dv(t,h).
\end{align}
Using \eqref{aste}, the dominated convergence theorem and pointwise convergence of $g_k$ and $g^\prime_k$, give
\begin{eqnarray*}
&&\hspace{+17pt}g_k\circ v\to g\circ v\quad\hspace{+35pt}\text{in $L^p(0,{T};L^p({\mathcal{H}}_w,\mu))$},\\[+8pt]
&&D(g_k\circ v)\to g^\prime(v)\,Dv\quad\hspace{+16pt}\text{in $L^p(0,{T};L^p({\mathcal{H}_w},\mu;{\mathcal{H}_w}))$},\\[+4pt]
&&\frac{\partial}{\partial t}(g_k\circ v)\to g^\prime(v)\,\frac{\partial}{\partial t}v\quad\hspace{+12pt}\text{in $L^2(0,{T};L^p({\mathcal{H}_w},\mu))$},
\end{eqnarray*}
with $1\le p<+\infty$. Recall that $D=\overline{D}$ and $\frac{\partial}{\partial t}$ are closed in $L^p(0,T;L^p(\mathcal{H}_w,\mu))$ and hence $D(g\circ v)=g^\prime(v)\,Dv$ and $\frac{\partial}{\partial t}(g\circ v)=g^\prime(v)\,\frac{\partial}{\partial t}v$.

If we set $\Psi_k:=g_k\circ v$, then $\Psi_k$ fulfils \eqref{lip-psi1}, \eqref{lip-psi2} and \eqref{conv-psi1} and \eqref{conv-psi2} hold.  Now the Lipschitz continuity \eqref{lip-psi1}, the mollifiers' properties and the obvious fact that $|(x-y)^+-(x)^+|\leq|y|$, imply
\begin{align}\label{gieffe}
\left|g_k(v(t,h))-g(v(t,h))\right|&\leq\int_{\mathbb{R}}{\rho_k(y)\left|g(v(t,h)-y)-g(v(t,h))\right|dy}
\leq\int_{\mathbb{R}}{\rho_k(y)\left|y\right|dy}\nonumber\\
&=\int_{[-\frac{1}{k},\frac{1}{k}]}{\rho_k(y)\left|y\right|dy}\leq\frac{1}{k}.
\end{align}
The uniform convergence of $\Psi_k$ to $\Psi$ as $k\to\infty$ follows from \eqref{gieffe}.
\end{proof}

We now associate an optimal stopping problem to each smooth function $\Psi_k$ and we denote by ${V}_k$ its value function. That is, we set
\begin{eqnarray}\label{sfinito}
{V}_k(t,h):=\sup_{t\leq\tau\leq T}\mathbb{E}\left\{e^{-\int_t^\tau{r^{t,h}_s(0)ds}}\,\Psi_k(\tau,r^{t,h}_\tau)\right\}.
\end{eqnarray}
It is not hard to verify that $V_k$ has the same regularity properties as $V$ (cf.~Theorem \ref{reg-v-fun}). We prove that ${V}_k$ is an approximation of the value function $V$ of \eqref{priceHJM2}.
\begin{prop}\label{regular}
The convergence ${V}_k\to V$ is uniform on bounded subsets $[0,T]\times\mathcal{B}\subset[0,T]\times{\mathcal{H}_w}$; that is, there exist positive constants $\gamma$ and $\beta$ such that
\begin{eqnarray}\label{unif-moll}
\sup_{(t,h)\in[0,T]\times{\mathcal{B}}}|{V}_k(t,h)-V(t,h)|\le\gamma e^{\beta\|h\|_w}\frac{1}{k},\quad\text{for all $k\in\mathbb{N}$}.
\end{eqnarray}
\end{prop}
\begin{proof}
Fix $(t,h)\in[0,T]\times{\mathcal{H}_w}$, then
\begin{align*}
V_k(t,h)-V(t,h)&\leq\sup_{t\leq\tau\leq T}\mathbb{E}\left\{e^{-\int_t^\tau{r^{t,h}_s(0)ds}}\left(\Psi_k(\tau,r^{t,h}_\tau)-\Psi(\tau,r^{t,h}_\tau)\right)
\right\}
\leq \gamma e^{\beta\|h\|_w}\frac{1}{k}
\end{align*}
by \eqref{uni-psi} and Lemma \ref{bound-discount}. The same holds for $V(t,h)-V_k(t,h)$ and hence \eqref{unif-moll} follows.
\end{proof}
\begin{prop}\label{rem-cont-Vk}
For every $k\in\mathbb{N}$ the map $(t,h)\mapsto V_k(t,h)$ is continuous on $[0,T]\times\mathcal{H}_w$.
\end{prop}
\begin{proof}
The proof follows from arguments as in the proof of Theorem \ref{reg-v-fun}.
\end{proof}

Set (cf.~\eqref{family})
\begin{align}\label{family2}
Y^{t,h}_k(\tau):=e^{-\int^{\tau}_t{r^{t,h}_u(0)du}}V_k(\tau,r^{t,h}_{\tau}),
\end{align}
where $\tau\in[t,T]$ is an arbitrary stopping time.
\begin{prop}\label{unif-intk}
For any $(t,h)\in[0,T]\times\mathcal{H}_w$ and any sequence $(\tau_k)_{k\in\mathbb{N}}$ of stopping times in $[t,T]$ the sequence $\big(Y^{t,h}_k(\tau_k)\big)_{k\in\mathbb{N}}$ is uniformly integrable.
\end{prop}
\begin{proof}
The proof follows along the same lines as in the proof of Proposition \ref{unif-int}.
\end{proof}

Notice that the approximating optimal stopping problem \eqref{sfinito} may be characterized through variational methods as in \cite{Ch-DeA12a}, apart for minimal adjustments. Once $V_k$ is found to be a solution of a suitable infinite-dimensional variational problem, 
taking the limit as $k\to\infty$ will lead to the variational formulation for the price function $V$. 
\section{An approximating variational inequality}

A variational formulation of problem \eqref{sfinito} is obtained by reducing the optimal stopping problem to a finite-dimensional setting. This is accomplished in two steps: first we make a Yosida approximation of the unbounded operator $A$ in \eqref{infiniteSDEb} by bounded operators $A_\alpha$, then we reduce the the SDE itself to a finite dimensional one. At each step a corresponding optimal stopping problem is studied. In order to proceed with this algorithm the SDE \eqref{infiniteSDEb} must live in a larger probability space. In particular, we assume that $W:=(W^0,W^1,W^2,\ldots)$ is an infinite sequence of real, standard Brownian motions on $(\Omega,\mathcal{F},\mathbb{P})$ and that the Brownian motion $B$ of \eqref{infiniteSDEb} coincides with its first component, i.e.~we set $W^0=B$. The original filtration can be replaced by the filtration generated by $W$, again denoted by $\{\mathcal{F}_t,t\geq0\}$ and completed by the null sets.

In this new setting all the arguments of the previous sections still hold and the pricing problem keeps the same form. In the next two sections we outline both the Yosida approximation and the finite-dimensional one. Full details may be found in \cite{Ch-DeA12a}.
\subsection{Yosida approximation}

The Yosida approximation of the unbounded linear operator $A$ may be introduced without any further assumption and it is defined by $A_\alpha:=\alpha A(\alpha I-A)^{-1}$, for $\alpha>0$ (cf.~\cite{Pazy}). Since $A_\alpha$ is a bounded linear operator, the corresponding SDE
\begin{equation}\label{SDE-yos}
\left\{
\begin{array}{l}
dr^{(\alpha)h}_t=[A_\alpha r^{(\alpha)h}_t+F_\sigma(r^{(\alpha)h}_t)]dt+\sigma(r^{(\alpha)h}_t)dW^0_t,\qquad t\in[0,T],\\
\\
r^{(\alpha)h}_0=h,
\end{array}
\right.
\end{equation}
admits a unique strong solution $r^{(\alpha)h}$. That is,
\begin{equation*}
r^{(\alpha)h}_t=h+\int_0^t{[A_\alpha r^{(\alpha)h}_s+F_\sigma(r^{(\alpha)h}_s)]ds}+\int_0^t{\sigma(r^{(\alpha)h}_s)dW^0_s},\qquad t\in[0,T],\,\mathbb{P}\textrm{-a.s.}
\end{equation*}
For each $\alpha>0$, the notations $r^{(\alpha)h}$ and $r^{(\alpha)t,h}$ are analogous to those used in Section \ref{sec-pricing}. The following important convergence result is proven in \cite{DaPr-Zab}, Proposition 7.5 and it is here recalled for completeness.
\begin{prop}\label{conv-yos}
Let $r^h$ be the unique mild solution of equation (\ref{infiniteSDEb}) with $r^h_0=h$ and let $r^{(\alpha)h}$ be the unique strong solution of equation (\ref{SDE-yos}). For $1\le p<\infty$, the following convergence holds
\begin{eqnarray*}
\lim_{\alpha\to\infty}{\mathbb{E}\left\{\sup_{0\leq t\leq T}\|r^{(\alpha)h}_t-r^h_t\|^p_{{w}}\right\}}=0,\qquad h\in{\mathcal{H}_w}.
\end{eqnarray*}
\end{prop}
\noindent
For $k\in\mathbb{N}$ arbitrary but fixed, define $V_{k,\alpha}$ as the value function of the optimal stopping problem corresponding to $r^{(\alpha)h}$, with regularized gain $\Psi_k$, 
\begin{eqnarray}\label{OS2}
V_{k,\alpha}(t,h):=\sup_{t\leq\tau\leq T}\mathbb{E}\left\{e^{-\int^{\tau}_t{r^{(\alpha)t,h}_s(0)ds}}\,\Psi_k(\tau,r^{(\alpha)t,h}_\tau)\right\}.
\end{eqnarray}
Notice that $V_{k,\alpha}$ satisfies \eqref{bdd-V}, \eqref{lip-V} and \eqref{lip-Vt} with the same constants. The convergence of $V_{k,\alpha}$ to $V_k$ (cf.~\eqref{sfinito}) as $\alpha\to\infty$ holds both uniformly with respect to $t$ and in suitable $L^p$-norms.
\begin{theorem}\label{cnv1}
For arbitrary $\alpha>0$ and $k\in\mathbb{N}$, the map $(t,h)\mapsto V_{k,\alpha}(t,h)$ is con\-ti\-nuous on $[0,T]\times\mathcal{H}_w$. Moreover, for $1\le p<+\infty$, the following convergence results hold
\begin{align}
&\sup_{0\leq t\leq T}|V_{k,\alpha}(t,h)-V_k(t,h)|\to 0\qquad\text{as $\alpha\to\infty$, $h\in{\mathcal{H}_w}$},\label{cv-1}\\[+6pt]
&V_{k,\alpha}\to V_k\qquad\text{as $\alpha\to\infty$, in $L^p(0,T;L^p(\mathcal{H}_w,\mu))$},\label{cv-gen01}\\[+8pt]
&V_{k,\alpha}\rightharpoonup V_k\qquad\text{as $\alpha\to\infty$, in $L^p(0,T;W^{1,p}(\mathcal{H}_w,\mu))$ (up to a subsequence),}\label{cv-gen02}\\[+6pt]
&\frac{\partial\,V_{k,\alpha}}{\partial\,t}\rightharpoonup \frac{\partial\,V_{k}}{\partial\,t}\qquad\text{as $\alpha\to\infty$, in $L^2(0,T;L^{p}(\mathcal{H}_w,\mu))$ (up to a subsequence).}\label{cv-gen03}
\end{align}
\end{theorem}
\begin{proof}
Continuity follows by arguments as in the proof of Theorem \ref{reg-v-fun}.

The Lipschitz property of the gain function $\Psi_k$ and the time-homogeneous character of the processes give
\begin{eqnarray*}
|V_{k,\alpha}(t,h)-V_k(t,h)|&\leq & 2\,L_V\,e^{\frac{\beta}{2}\|h\|_w}\mathbb{E}\bigg\{\sup_{0\leq s\leq T}\left\|r^{(\alpha)h}_s-r^{h}_s\right\|^2_{w}\bigg\}^{\frac{1}{2}},
\end{eqnarray*}
by arguments similar to those used in the proof of Theorem \ref{reg-v-fun}. Since $L_V$ is independent of $t$, the uniform convergence \eqref{cv-1} follows from Proposition \ref{conv-yos}. Then, dominated convergence theorem and \eqref{cv-1} imply \eqref{cv-gen01} as the same bounds of $V$ in \eqref{bdd-V} hold for $V_{k,\alpha}$.

As for \eqref{cv-gen02} and \eqref{cv-gen03} we notice that bounds \eqref{lip-V} and \eqref{lip-Vt} hold for $V_{k,\alpha}$ with the same constants. Therefore, $V_{k,\alpha}$ and $\frac{\partial}{\partial\,t}V_{k,\alpha}$ are functions bounded in $L^p(0,T;W^{1,p}(\mathcal{H}_w,\mu))$ and $L^2(0,T;L^{p}(\mathcal{H}_w,\mu))$, respectively, uniformly with respect to $\alpha>0$ and $k\in\mathbb{N}$ (cf.~Corollary \ref{arto}). Hence one can extract weakly converging subsequences and \eqref{cv-gen02} and \eqref{cv-gen03} follow by uniqueness of the limit and by \eqref{cv-gen01}.
\end{proof}
\begin{theorem}\label{ascoli}
For every $k\in\mathbb{N}$ the convergence $V_{k,\alpha}\to V_k$ as $\alpha\to\infty$ is uniform on compact subsets $[0,T]\times\mathcal{K}\subset [0,T]\times{\mathcal{H}_w}$.
\end{theorem}
\begin{proof}
We only outline the proof as a similar result is proved in \cite{Ch-DeA12a}, Corollary 3.3. Fix $h\in{\mathcal{H}_w}$ and for each $\alpha>0$ define
\begin{equation*}
F_\alpha(h):=\sup_{t\in[0,T]}|V_{k,\alpha}(t,h)-V_k(t,h)|.
\end{equation*}
Then $F_\alpha(h)\to 0$ as $\alpha\to\infty$ by \eqref{cv-1}. One may verify that the family $(F_\alpha)_{\alpha>0}$ is equi-bounded and equi-continuous on bounded subsets of $[0,T]\times\mathcal{H}_w$, as \eqref{bdd-V} and \eqref{lip-V} hold for both $V_{k,\alpha}$ and $V_k$,
hence $V_{k,\alpha}$ converges uniformly to $V_k$, as $\alpha\to\infty$, on compact subsets $[0,T]\times\mathcal{K}$ (\cite{Die}, Theorem 7.5.6). 
\end{proof}
\subsection{Finite dimensional reduction}

For each $n\in\mathbb{N}$ consider the finite dimensional subset $\mathcal{H}^{(n)}_w:=span\{\varphi_1,\varphi_2,\ldots,\varphi_n\}$ and the orthogonal projection operator $P_n:\mathcal{H}_w\to\mathcal{H}^{(n)}_w$. Approximate the diffusion coefficients of \eqref{SDE-yos} by $\sigma^{(n)}:=(P_n\sigma)\circ P_n$, $F^{(n)}_\sigma:=(P_nF_\sigma)\circ P_n$ and $A_{\alpha,n}:=P_nA_\alpha P_n$, respectively. Notice that $A_{\alpha,n}$ is a bounded linear operator on $\mathcal{H}^{(n)}_w$. Define the process $r^{(\alpha)h;n}$ as the unique strong solution of the SDE on $\mathcal{H}^{(n)}_w$ given by
\begin{eqnarray}\label{SDE-n}
\left\{
\begin{array}{l}
dr^{(\alpha)h;n}_t=[A_{\alpha,n}r^{(\alpha)h;n}_t+F^{(n)}_\sigma(r^{(\alpha)h;n}_t)]dt+
\sigma^{(n)}(r^{(\alpha)h;n}_t)dW^0_t+
\epsilon_n\sum^n_{i=1}{\varphi_i\,dW^i_t},\\
\\
r^{(\alpha)h;n}_0=P_nh=:h^{(n)},\qquad\:\: t\in[0,T],
\end{array}
\right.
\end{eqnarray}
where $(\epsilon_n)_{n\in\mathbb{N}}$ is a sequence of positive numbers such that
\begin{align}\label{epsilon}
\sqrt{n}\,\epsilon_n\to0\qquad \textrm{as}\:\: n\to\infty.
\end{align}
\begin{prop}\label{c-1}
The following convergence
\begin{equation}\label{limit1}
\lim_{n\rightarrow\infty}\mathbb{E}\left\{\sup_{t\in[0,T]}\left\|r^{(\alpha)h;n}_t-r^{(\alpha)h}_t\right\|^2_w\right\}=0,
\end{equation}
holds uniformly with respect to $h$ on compact subsets of $\mathcal{H}_w$.
\end{prop}
\begin{proof}
The proof is based on standard $L^p$-estimates of SDE's strong solutions. It follows along the same lines as the proof of \cite{Ch-DeA12a}, Proposition 3.5. In fact, the only difference here is the presence of a non-linear drift term in \eqref{SDE-yos} and \eqref{SDE-n} which, however, may be estimated by using \eqref{lipF} and \eqref{b-F}.
\end{proof}
\begin{remark}\label{unif}
Notice that, for any starting time $t\in[0,T]$, the previous proposition and the arguments of its proof hold for $r^{(\alpha)t,h;n}$ and $r^{(\alpha)t,h}$ as well, thanks to the time-homogeneity of equations \eqref{SDE-yos} and \eqref{SDE-n}.
\end{remark}

\begin{lemma}\label{bound-discount-n}
There exist positive constants $\gamma^\prime$ and $\beta^\prime$, independent of $\alpha$ and $n$, such that the bound \eqref{bound-r0} holds for $r^{(\alpha)t,h;n}$ with $\beta$ and $\gamma$ replaced by $\beta^\prime$ and $\gamma^\prime$, respectively.
\end{lemma}
\begin{proof}
As in the proof of Lemma \ref{bound-discount} it suffices to consider $t^\prime=0$ and $p=1$. For $h\in\mathcal{H}_w$ we simplify notation by setting $r^n_t:=r^{(\alpha)h;n}_t$, $t\ge0$ and by taking $B(h):=\exp\big\{CT\|h\|_w+CC_FC^2_\sigma T^2\big\}$. Following the same steps as in the proof of Lemma \ref{bound-discount} and recalling that $(W^i_t)_{i\in\mathbb{N}}$ is a family of independent Brownian motions we find
\begin{align*}
\mathbb{E}&\Big\{\sup_{0\le v\le T}e^{-\int_0^v{r^{n}_t(0)dt}}\Big\}\\
\le&B(h)\Big(\frac{1}{T}\hspace{-4pt}\int^T_0{\hspace{-6pt}\mathbb{E}\Big\{e^{-T\int^t_0{S(t-s)\sigma^{(n)}(r^n_s)(0)dW^0_s}}
\prod^n_{i=1}e^{-\epsilon_n T\int^t_0{S(t-s)\varphi_i(0)dW^i_s}}\Big\}dt}+1\Big)\nonumber\\
\le&B(h)\Big(\frac{1}{T}\hspace{-4pt}\int^T_0{\hspace{-6pt}\mathbb{E}\Big\{e^{-2T\int^t_0{S(t-s)\sigma^{(n)}(r^n_s)(0)dW^0_s}}
\Big\}^{\frac{1}{2}}\Big(\prod^n_{i=1}
\mathbb{E}\Big\{e^{-2\epsilon_n T\int^t_0{S(t-s)\varphi_i(0)dW^i_s}}\Big\}\Big)^{\frac{1}{2}}dt}+1\Big).\nonumber
\end{align*}
The first expectation inside the time-integral above may be estimated as in the proof of Lemma \ref{bound-discount}. As for each term of the infinite product one has
\begin{align*}
\mathbb{E}\Big\{e^{-\epsilon_n T\int^t_0{S(t-s)\varphi_i(0)dW^i_s}}\Big\}\le e^{{2}C\,T^3\epsilon^2_n}
\end{align*}
with $C>0$ as in \eqref{injection}, by using arguments similar to those employed to obtain \eqref{bound-dis5} and by recalling that $\|\varphi_i\|_w=1$, $i\in\mathbb{N}$. Therefore, it follows that
\begin{align*}
\prod^n_{i=1}
\mathbb{E}\Big\{e^{-\epsilon_n T\int^t_0{S(t-s)\varphi_i(0)dW^i_s}}\Big\}\le e^{{2}C\,T^3n\,\epsilon^2_n}.
\end{align*}
Hence, the rate of convergence to zero of $\epsilon_n$ (cf.~\eqref{epsilon}) enables us to pick constants $\gamma^\prime$ and $\beta^\prime$ large enough to guarantee that the bound \eqref{bound-r0} holds for $r^{(\alpha)t,h;n}$ with $\gamma$ and $\beta$ replaced by $\gamma^\prime$ and $\beta^\prime$, respectively, uniformly in $\alpha$ and $n$.
\end{proof}

For $n\geq1$ define $\Psi^{(n)}_k:[0,T]\times\mathcal{H}_w\to\mathbb{R}$ by
\begin{equation}
\Psi^{(n)}_k(t,h):=\Psi_k(t,P_nh)=\Psi_k(t,h^{(n)})
\end{equation}
(cf.~\eqref{SDE-n}). 
Of course, $P_nh^{(n)}=h^{(n)}$, hence $\Psi^{(n)}_k(t,\,\cdot\,)=\Psi_k(t,\,\cdot\,)$ on $\mathcal{H}^{(n)}_w$. However, in what follows it is convenient to use the notation $\Psi^{(n)}_k$ since this is interpreted as a gain function on $[0,T]\times\mathcal{H}^{(n)}_w$. 
The same arguments as in Appendix \ref{app0}, \eqref{psi-n2} show that $\Psi^{(n)}_k\to\Psi_k$ as $n\to\infty$ uniformly on every compact subset of $[0,T]\times\mathcal{H}_w$.
Let $V^{(n)}_{k,\alpha}$ be the value function of the optimal stopping problem
\begin{equation}\label{OS3}
V_{k,\alpha}^{(n)}(t,h^{(n)}):=\sup_{t\leq\tau\leq T}\mathbb{E}\left\{e^{-\int^{\tau}_t{r^{(\alpha)t,h;n}_s(0)ds}}\,\Psi^{(n)}_k(\tau,r^{(\alpha)t,h;n}_\tau)\right\},
\end{equation}
of course, $V_{k,\alpha}^{(n)}$ may also be seen as a function on $[0,T]\times\mathbb{R}^{n}$. Notice that, as for $V_{k,\alpha}$, again $V_{k,\alpha}^{(n)}$ satisfies \eqref{bdd-V}, \eqref{lip-V} and \eqref{lip-Vt} with the same constants (cf.~Lemma \ref{bound-discount-n}). Results similar to Theorem \ref{cnv1} and Theorem \ref{ascoli} hold.
\begin{theorem}\label{cnv2}
For any $\alpha>0$, $k\in\mathbb{N}$ and $n\in\mathbb{N}$, the map $(t,h)\mapsto V^{(n)}_{k,\alpha}(t,h^{(n)})$ is continuous on $[0,T]\times\mathcal{H}_w$. Moreover, for $1\le p<+\infty$, the following convergence results hold,
\begin{align}
&\sup_{(t,h)\in[0,T]\times\mathcal{K}}|V^{(n)}_{k,\alpha}(t,h^{(n)})-V_{k,\alpha}(t,h)|\to 0\quad\text{as $n\to\infty$, $\mathcal{K}\subset\mathcal{H}_w$, $\mathcal{K}$ compact},\label{cv-3}\\[+6pt]
&V^{(n)}_{k,\alpha}\to V_{k,\alpha}\qquad\text{as $n\to\infty$, in $L^p(0,T;L^p(\mathcal{H}_w,\mu))$},\label{cv-gen04}\\[+8pt]
&V^{(n)}_{k,\alpha}\rightharpoonup V_{k,\alpha}\qquad\text{as $n\to\infty$, in $L^p(0,T;W^{1,p}(\mathcal{H}_w,\mu))$ (up to a subsequence),}\label{cv-gen05}\\[+6pt]
&\frac{\partial\,V^{(n)}_{k,\alpha}}{\partial\,t}\rightharpoonup \frac{\partial\,V_{k,\alpha}}{\partial\,t}\qquad\text{as $n\to\infty$, in $L^2(0,T;L^{p}(\mathcal{H}_w,\mu))$ (up to a subsequence).}\label{cv-gen06}
\end{align}
\end{theorem}
\begin{proof}
The proof of \eqref{cv-gen04} follows along the same lines as the proof of \eqref{cv-gen01}. Similarly the uniform convergence in Proposition \ref{c-1} implies \eqref{cv-3}. As in Theorem \ref{cnv1} weak convergences \eqref{cv-gen05} and \eqref{cv-gen06} follow from \eqref{lip-V} and \eqref{lip-Vt}. Finally, continuity is obtained by means of arguments as those used in the proof of Theorem \ref{reg-v-fun}.
\end{proof}

Uniform integrability conditions hold for $V^{(n)}_{k,\alpha}$ and $V_{k,\alpha}$ as in Proposition \ref{unif-int} for $V_k$. For fixed $k\in\mathbb{N}$ set
\begin{align}
X^{t,h}_{\alpha}(\sigma):=e^{-\int^{\sigma}_t{r^{(\alpha)t,h}_u(0)du}}V_{k,\alpha}(\sigma,
r^{(\alpha)t,h}_{\sigma}),\qquad\alpha>0,
\end{align}
and for fixed $k\in\mathbb{N}$ and $\alpha>0$ set
\begin{align}
Z^{t,h}_{n}(\tau):=e^{-\int^{\tau}_t{r^{(\alpha)t,h;n}_u(0)du}}V^{(n)}_{k,\alpha}(\tau,r^{(\alpha)t,h;n}_{\tau}),\qquad n\in\mathbb{N},
\end{align}
where $\sigma$ and $\tau$ are arbitrary stopping times in $[t,T]$ (cf.~\eqref{family} and \eqref{family2}).
\begin{prop}\label{unif-intk-na}
For any $(t,h)\in[0,T]\times\mathcal{H}_w$, for any family $(\sigma_\alpha)_{\alpha>0}$ and any sequence $(\tau_n)_{n\in\mathbb{N}}$ of stopping times in $[t,T]$, the family $\big(X^{t,h}_{\alpha}(\sigma_\alpha)\big)_{\alpha>0}$ and the sequence $\big(Z^{t,h}_n(\tau_n)\big)_{n\in\mathbb{N}}$ are uniformly integrable.
\end{prop}
\subsection{A variational inequality for $V_k$}
In order to obtain a well-posed infinite dimensional variational inequality for $V_k$ we need to introduce the two assumptions below on the diffusion coefficient $\sigma$ and on the trace class operator $Q$ (cf.~Definition \ref{Q-op}) which will be standing assumptions in the rest of the paper.
\begin{ass}\label{ass-sigma}
The map $\sigma:\mathcal{H}_w\to\mathcal{H}^0_w$ is such that
\begin{itemize}
\item [\textbf{(1)}] $\sigma(h)\in Q(\mathcal{H}_w)$, $\forall h\in\mathcal{H}_w$; i.e., there exists $\vartheta:\mathcal{H}_w\to\mathcal{H}^0_w$ such that $\sigma(h)= Q\vartheta(h)$.
\item [\textbf{(2)}] $\vartheta$ is such that $h\mapsto D\vartheta(h)$ is continuous and bounded on $\mathcal{H}_w$ (cf.~\eqref{Dfried} and \eqref{Dfried02}).
\end{itemize}
\end{ass}

Denote $A^*$ the adjoint operator of $A$. We make the following assumption linking the operators $Q$ and $A$.
\begin{ass}\label{ass-Q}
The operator $Q$ of Definition \ref{Q-op} is such that
\begin{equation}\label{rate}
Tr\big[\,A\,Q\,A^*\,\big]<\infty.
\end{equation}
\end{ass}
\noindent Although the choice of $Q$ is now subjected to Assumption \ref{ass-Q}, Assumption \ref{ass-sigma} still allows enough generality for the volatility structure. For instance all deterministic volatility structures $\sigma(r_t)(x)=\sigma(x)$ as well as all finite dimensional volatility structures fit into this setting under a suitable choice of $Q$. Also, whenever {there exists $\ell:\mathcal{H}_w\to \mathbb{R}$ such that $\lim_{x\to\infty}Q\vartheta(h)(x)=\ell(h)\neq 0$} (the limit always exists), one finds a volatility $\sigma$ satisfying Assumption \ref{ass-sigma} with $\tilde{\vartheta}(h)(\,\cdot\,)=\vartheta(h)(\,\cdot\,)-\ell(h) Q^{-1}$.

We start by characterizing $V_k$ as a solution of a suitable infinite dimensional obstacle problem.
For that we initially characterize $V^{(n)}_{k,\alpha}$ as a solution of a variational problem on $[0,T]\times\mathbb{R}^n$ by means of a slight modification of standard results; then, we take limits as $n\to\infty$ and $\alpha\to\infty$ in the finite-dimensional variational problem and use Theorems \ref{cnv1} and \ref{cnv2} to obtain a variational inequality for $V_k$.
This methodology was developed in \cite{Ch-DeA12a} but in a different setting and here we omit some of the details and recall results contained in \cite{Ch-DeA12a} when possible.

We start by introducing the infinitesimal generator $\mathcal{L}$ of the diffusion $r$. Let $C^2_{b,F}(\mathcal{H}_w;\mathbb{R})$ be the set of bounded continuous functions which are twice continuously differentiable in the Frech\'et sense with bounded derivatives. Then for every $g\in C^2_{b,F}(\mathcal{H}_w;\mathbb{R})$ with $Dg$ taking values in $D(A^*)$ one has
\begin{equation}\label{ellal}
\mathcal{L}g(h)=\frac{1}{2}Tr\left[\sigma\sigma^*(h)D^2g(h)\right]+\langle h,A^* Dg(h)\rangle_w+\langle F_\sigma(h),Dg(h)\rangle_w,\hspace{+20pt}\textit{for}\,\,h\in\mathcal{H}_w.
\end{equation}
Take $1<p<\infty$, let $p'$ be such that $\frac{1}{p}+\frac{1}{p'}=1$ and define the space
\begin{equation}\label{set}
\mathcal{V}^p:=\{v\,|\,v\in L^{2p}(\mathcal{H}_w,\mu)\:\textrm{and}\:\:D v\in L^{2p'}(\mathcal{H}_w,\mu;\mathcal{H}_w)\},
\end{equation}
endowed with the norm
\begin{eqnarray}
\lnr v\rnr_{p}:=\|v\|_{L^{2p}(\mathcal{H}_w,\mu)}+\|D v\|_{L^{2p'}(\mathcal{H}_w,\mu;\mathcal{H}_w)}.
\end{eqnarray}
Then $\big(\mathcal{V}^p,\lnr\,\cdot\,\rnr_p\big)$ is a separable Banach space. Since $D(A^*)$ is dense in $\mathcal{H}_w$, the set
\begin{equation}\label{dense}
\mathcal{E}_A(\mathcal{H}_w):=\textrm{span}\{\mathscr{R}e(\phi_{g}),\,\mathscr{I}m
(\phi_{g}),
\:\phi_{g}(h)=e^{i\langle g,h\rangle_{w}},\,g\in D(A^*)\}
\end{equation}
is dense in $\mathcal{V}^p$ (cf.~\cite{DaPr}, Chapter 10 and \cite{DaPr-Zab04}, Chapter 9).

In the spirit of \cite{Ben-Lio82} we will associate the second order differential operator in \eqref{ellal} to a bilinear form on $\mathcal{V}^p$. Let us first analyze the second term on the right-hand side of \eqref{ellal}. For every $u\in\mathcal{E}_A(\mathcal{H}_w)$ it is easy to check that $A^*Du\in L^p(\mathcal{H}_w,\mu;\mathcal{H}_w)$, $1<p<\infty$. Hence, for $v\in \mathcal{V}^p$, we define
\begin{eqnarray}\label{rab3}
L_{A}(v,u):=\int_{\mathcal{H}_w}\langle h,A^*D u\rangle_{w} v \,\mu(dh),\quad\textrm{for}\:\:u\in\mathcal{E}_A(\mathcal{H}_w).
\end{eqnarray}
By Assumption \ref{ass-Q} it was shown in \cite{Ch-DeA12a}, Section 4.3.2 (see eq.~(4.93) therein) that
\begin{eqnarray}\label{ext}
|L_{A}(v,u)|\leq Tr\big[\,A\,Q\,A^*\,\big] \lnr v\rnr_p\,\lnr u\rnr_p,\quad u\in\mathcal{E}_A(\mathcal{H}_w).
\end{eqnarray}
As $\mathcal{E}_A(\mathcal{H}_w)$ is dense in $\mathcal{V}^p$, $L_{A}(v,\cdot)$ is extended to the whole space $\mathcal{V}^p$ 
and the extended functional is denoted by $\bar{L}_{A}(v,\cdot)$.

For $u,v\in \mathcal{V}^p$ define the bilinear form
\begin{align}\label{bl-form}
a_{\mu}(u,v):=&\frac{1}{2}\int_{\mathcal{H}_w}{\langle\sigma\sigma^{*}Du,Dv\rangle_{w} \mu(dh)}+\int_{\mathcal{H}_w}{\langle\hat{C},Du\rangle_{w} v \, \mu(dh)}\nonumber\\
&-\bar{L}_{A}(v,u)+\int_{\mathcal{H}_w}{h(0)u\,v\mu(dh)},
\end{align}
where $\hat{C}:\mathcal{H}_w\to \mathcal{H}_w$ is defined by
\begin{align*}
\hat{C}(h):=\frac{1}{2}\big(Tr[D\sigma]\sigma+D\sigma\cdot\sigma-2F_\sigma\big)(h)-\frac{1}{2}\sigma\sigma^*(h)Q^{-1}h
\end{align*}
with $D\sigma\cdot\sigma$ denoting the action of $D\sigma\in\mathcal{L}(\mathcal{H}_w)$ on the vector $\sigma\in\mathcal{H}_w$.

The following estimate holds.
\begin{theorem}\label{uni-est}
For every $1< p<\infty$ there exists a constant $C_{\mu,\vartheta,p}>0$, depending on $\mu$, $p$ and the bounds of $\vartheta$ in Assumption \ref{ass-sigma}, such that
\begin{equation}\label{univ}
\int_0^T{\hspace{-5pt}|a_\mu(u(t),v(t))| dt}\leq C_{\mu,\vartheta,p}\left(\int_0^T{\lnr u(t)\rnr^2_p}dt\right)^\frac{1}{2}\left(\int_0^T{\lnr v(t)\rnr^2_p}dt\right)^\frac{1}{2}
\end{equation}
for all $u,v\in L^2(0,T;\mathcal{V}^p)$.
\end{theorem}
\begin{proof}
Except for the last term in \eqref{bl-form} the proof is analogous to the one of \cite{Ch-DeA12a}, Theorem 4.9 (we remark that the bound in \cite{Ch-DeA12a} is uniform in $(n,\alpha)\in\mathbb{N}\times\mathbb{R}_+$) and it follows from Assumption \ref{ass-sigma} and bounds \eqref{b-F} and \eqref{ext}. As for the last term, from \eqref{injection} and H\"older inequality we obtain
\begin{align*}
\Big|\int_0^T\int_{\mathcal{H}_w}{h(0)u\,v\mu(dh)dt}\Big|&\leq C\hspace{-2pt}\int_0^T\hspace{-6pt}\Big[\Big(\int_{\mathcal{H}_w}{\hspace{-4pt}\|h\|^2_w|u(t,h)|^2\mu(dh)}\Big)
^{\frac{1}{2}}\big\|v(t)\big\|_{L^2(\mathcal{H}_w,\mu)}
\Big]dt
\end{align*}
and \eqref{univ} easily follows by another application of H\"older's inequality and by observing that $\mathcal{V}^p\hookrightarrow L^2(\mathcal{H}_w,\mu)$.
\end{proof}
\begin{remark}\label{buio}
For functions $u,v$ in $C^{2}_b(\mathcal{H}_w;\mathbb{R})$ in the Friedrichs' sense (cf.~\eqref{Dfried} and \eqref{Dfried03})  with $Du$ taking values in $D(A^*)$, \eqref{bl-form} is simply obtained by
\begin{align}\label{bl-form2}
a_{\mu}(u,v):=-\int_{\mathcal{H}_w}{\Big(\mathcal{L}u-h(0)u\Big)\,v\,\mu(dh)}.
\end{align}
In fact, for $u_n:=u\circ P_n$, $v_n:=v\circ P_n$ defined on $n$-dimensional subspaces of $\mathcal{H}_w$ and $\sigma^{(n)}$, $F^{(n)}_\sigma$ as in \eqref{SDE-n}, \eqref{bl-form2} follows by Green's formula and \eqref{ellal}. Then, by taking the limit as $n\to\infty$ and using dominated convergence one finds that \eqref{bl-form2} holds in general.
\end{remark}

Denote by $(\cdot,\cdot)_\mu$ the scalar product in $L^2(\mathcal{H}_w,\mu)$ and define $F_k(\,\cdot\,)(t)\in \mathcal{V}^{p\,*}$, $t\in[0,T]$ by
\begin{equation}\label{ruiz0}
\hspace{-5pt}F_k(v)(t)\hspace{-2pt}:=\hspace{-4pt}\Big(\frac{\partial\Psi_k}{\partial t}(t)\hspace{-1pt}+\hspace{-1pt}\frac{1}{2}Tr\big[\sigma\sigma^*D^2\Psi_k(t)\big]\hspace{-2pt}+\hspace{-1pt}\langle F_\sigma,D\Psi_k(t)\rangle_w-h(0)\Psi_k(t),v\Big)_\mu \hspace{-2pt}+\bar{L}_{A}(v,\Psi_k(t)),
\end{equation}
for all $v\in \mathcal{V}^p$. Take $1< p< \infty$ and introduce the closed, convex set
\begin{equation}\label{ccs2}
\mathcal{K}^p_{\mu}:=\big\{w:\,w\in \mathcal{V}^p\:\:\text{and}\:\:w\geq0\:\:\mu\textrm{-a.e.}\big\}.
\end{equation}

The next theorem was proved in \cite{Ch-DeA12a} (cf.~Theorem 4.22 and Theorem 4.24) and it was the main result of that paper. The extension to the present setting (which includes an unbounded stochastic discount factor) follows from the convergence results in Propositions \ref{conv-yos} and \ref{c-1} for the approximating diffusions $r^{(\alpha)t,h}$ and $r^{(\alpha)t,h;n}$, from the convergence results in Theorems \ref{cnv1}, \ref{ascoli} and \ref{cnv2} for the value functions $V_{k,\alpha}^{(n)}$ and $V_{k,\alpha}$, and from Theorem \ref{uni-est}. 

\begin{theorem}\label{inf-yos-vi}
For every $1< p<\infty$ the function $u_k:=V_k-\Psi_k$ is a solution of the variational problem
\begin{eqnarray}\label{d-4}
\left\{
\begin{array}{l}
u\in L^2(0,T;\mathcal{V}^p)\cap C([0,T]\times\mathcal{H}_w);\:\:\:\frac{\partial\,u}{\partial\,t}\in L^2(0,T;L^2(\mathcal{H}_w,\mu));\\[+10pt]
u(T,h)=0,\:h\in\mathcal{H}_w;\:\:\:\:u(t,h)\geq 0,\: (t,h)\in[0,T]\times\mathcal{H}_w;\\[+10pt]
-\big(\frac{\partial u}{\partial t}(t),v-u(t)\big)_{\mu}+a_\mu(u(t),v-u(t))-F_k\big(v-u(\,\cdot\,)\big)(t)\geq0\\[+5pt]
\hspace{+220pt}\textrm{for a.e.~$t\in[0,T]$ and all $v\in \mathcal{K}^p_{\mu}$.}
\end{array}\right.
\end{eqnarray}
The stopping time
\begin{equation}\label{rosp}
\tau^*_{k}(t,h):=\inf\{s\geq t\,:\,V_k(s,r^{t,h}_s)=\Psi_k(s,r^{t,h}_s)\}\wedge T
\end{equation}
is optimal for $V_k$ in \eqref{sfinito},
and
\begin{align}\label{sbruff}
V_k(t,h)=\mathbb{E}\left\{e^{-\int^\tau_t{r^{t,h}_u(0)du}}V_k(\tau,r^{t,h}_\tau)\right\}\quad\text{for any stopping time $\tau\leq\tau^*_{k}(t,h)$}.
\end{align}
\end{theorem}
\begin{proof}
We only outline the proof here and details can be found in \cite{Ch-DeA12a}. It goes through 4 steps.\vspace{+5pt}

\noindent\emph{Step 1: Finite dimensional bounded domains.} Fix $(\alpha,n)\in\mathbb{R}_+\times\mathbb{N}$. Take $R>0$ and let $B_R$ be the open ball in $\mathbb{R}^n$ with center in the origin and radius $R$. Let $\mathcal{L}_{\alpha,n}$ be the infinitesimal generator associated to $r^{(\alpha)x;n}$ of \eqref{SDE-n} and denote $f^{(\alpha,n)}_k:=\frac{\partial}{\partial\,t}\Psi^{(n)}_k+\mathcal{L}_{\alpha,n}\Psi^{(n)}_k-h^{(n)}(0)\Psi^{(n)}_k$. Since $\Psi^{(n)}_k\in C^{1,2}_b(\mathbb{R}^n;\mathbb{R})$ (cf.~Proposition \ref{mollify1}) it is well known (cf.~\cite{Fri82}, Theorem 8.2, p.~77; see~also \cite{Ch-DeA12a}, Proposition 4.2 and Corollary 4.3) that there exists a unique function $\bar{u}$ such that $\bar{u}\in L^p(0,T;W^{1,p}_0(B_R))\cap L^p(0,T;W^{2,p}(B_R))$, $\displaystyle{\frac{\partial\bar{u}}{\partial t}}\in L^p(0,T;L^{p}(B_R))$ for all $1\leq p<\infty$, and $\bar{u}\in C([0,T]\times\overline{B}_R)$ that solves the obstacle problem
\begin{eqnarray}\label{obs-1}
\left\{
\begin{array}{ll}
\max\left\{\displaystyle{\frac{\partial u}{\partial t}}+\mathcal{L}_{\alpha,n}u-h^{(n)}(0)u+f^{(\alpha,n)}_k\,,\,-u\right\}= 0\:\:\text{a.e.~in $(0,T)\times B_R$}; &\\[+12pt]
u(t,h^{(n)})\geq0\:\:\textrm{in}\:\:[0,T]\times \overline{B}_R\:\:\:\:\text{and}\:\:\:\: u(T,h^{(n)})=0,\:\: h^{(n)}\in\overline{B}_R;&\\[+10pt]
u(t,h^{(n)})=0,\:\: \text{in $[0,T]\times{\partial B_R}$}. &\\
\end{array}
\right.
\end{eqnarray}
Moreover, if $V^{(n)}_{k,\alpha,R}:=\big(\bar{u}+\Psi^{(n)}_k\big)$ then a verification argument based on a generalized It\^o's formula gives
\begin{align}\label{osbdd}
V^{(n)}_{k,\alpha,R}(t,h^{(n)})=\sup_{t\leq\tau\leq T}\mathbb{E}\left\{e^{-\int^{\tau_R\wedge\tau}_t{r^{(\alpha)t,h;n}_s(0)ds}}\,\Psi^{(n)}_k(\tau_R\wedge\tau
,r^{(\alpha)t,h;n}_{\tau_R\wedge\tau})\right\}
\end{align}
where $\tau_R:=\tau_R(t,h)$ is the first exit time of $r^{(\alpha)t,h;n}$ from $B_R$ (cf.~for instance proof of \cite{Ch-DeA12a}, Corollary 4.3). The same argument also provides the optimality of the stopping time
\begin{align}\label{ost00}
\tau^*_{k,\alpha,n,R}(t,h):=\inf\{s\geq t\,:\,V^{(n)}_{k,\alpha,R}(s,r^{(\alpha)t,h;n}_s)=\Psi^{(n)}_k(s,r^{(\alpha)t,h;n}_s)\}\wedge\tau_R\wedge T
\end{align}
and the analogous of \eqref{sbruff} holds for $V^{(n)}_{k,\alpha,R}$ with any stopping time $\tau\le\tau^*_{k,\alpha,n,R}$. Notice that the stochastic discount rate in \eqref{osbdd} is bounded in $B_R$ and the arguments used to prove existence and uniqueness of the solution to \eqref{obs-1} are the same as for the undiscounted problem (cf.~\cite{Ch-DeA12a}).\vspace{+5pt}

\noindent\emph{Step 2: Strong variational formulation for finite-dimensional bounded domains.} We denote by $a^{(\alpha,n)}_\mu(\,\cdot\,,\,\cdot\,)$ the bilinear form associated to $\mathcal{L}_{\alpha,n}$ (cf.~\eqref{bl-form2}) and by $F^{(\alpha,n)}_k$ the operator analogous to that in \eqref{ruiz0} but with $\Psi_k$, $\sigma$ and $A$ replaced by $\Psi^{(n)}_k$, $\sigma^{(n)}$ and $A_{\alpha,n}$, respectively. Set $\mathcal{K}^p_{n,\mu}:=\big\{w:\mathbb{R}^n\to\mathbb{R}\:\text{such that}\:w\in\mathcal{K}^p_\mu\big\}$ (cf.~\eqref{ccs2}). Then, by arguments as in \cite{Ch-DeA12a} (see the proof of Corollary 4.3 and of eq.~(4.43) therein) one can show that the obstacle problem \eqref{obs-1} may be equivalently written in the form of \eqref{d-4} but with $a_\mu(\,\cdot\,,\,\cdot\,)$ and $F_k$ replaced by $a^{(\alpha,n)}_\mu(\,\cdot\,,\,\cdot\,)$ and $F^{(\alpha,n)}_k$, respectively, and with $(\,\mathcal{V}^p_n\,,\,\mathcal{K}^p_{n,\mu}\,,\,\mathcal{H}^{(n)}_w\,)$ instead of $(\,\mathcal{V}^p\,,\,\mathcal{K}^p_{\mu}\,,\,\mathcal{H}_w\,)$. Clearly $\bar{u}:=V^{(n)}_{k,\alpha,R}- \Psi^{(n)}_k$ is the unique solution of the variational inequality in such new form and it is often referred to as \emph{strong solution}.\vspace{+5pt}

\noindent\emph{Step 3: Taking limits in the variational formulation.} The key fact is that the bounds in Theorem \ref{reg-v-fun} hold for $V^{(n)}_{k,\alpha,R}$ as well. It follows that $V^{(n)}_{k,\alpha,R}$ converges to $V^{(n)}_{k,\alpha}$ as $R\to\infty$, uniformly on compact subsets of $[0,T]\times\mathbb{R}^n$, strongly in $L^p(0,T;L^p(\mathbb{R}^n,\mu_n))$, and weakly in $L^p(0,T;W^{1,p}(\mathbb{R}^n,\mu_n))$ for $1\le p<\infty$. Moreover $\frac{\partial}{\partial\,t}V^{(n)}_{k,\alpha,R}\rightharpoonup \frac{\partial}{\partial\,t}V^{(n)}_{k,\alpha}$ in $L^2(0,T;L^p(\mathbb{R}^n,\mu_n))$ for $1\le p<\infty$ as $R\to\infty$. The proof relies on the same arguments as in the proof of Theorem \ref{cnv2}.

These convergence results and Theorems \ref{cnv1}, \ref{ascoli} and \ref{cnv2} allow us to take limits in the variational inequality as $R\to\infty$, $n\to\infty$ and $\alpha\to\infty$ in the prescribed order. The continuity property of the bilinear form needed in the limits to obtain the variational inequality \eqref{d-4} for $u_k$ is provided by Theorem \ref{uni-est}. Details are omitted here and may be found in \cite{Ch-DeA12a}, specifically in Theorems 4.12, 4.16 and 4.22. Some of these arguments will also be illustrated in the proof of Theorem \ref{vi-fin} below.\vspace{+5pt}

\noindent\emph{Step 4: Optimality of the stopping time \eqref{rosp}.} We denote
\begin{align}\label{ost01}
\tau^*_{k,\alpha,n}(t,h)&:=\inf\{s\geq t\,:\,V^{(n)}_{k,\alpha}(s,r^{(\alpha)t,h;n}_s)=\Psi^{(n)}_k(s,r^{(\alpha)t,h;n}_s)\}\wedge T, \\[+5pt]
\tau^*_{k,\alpha}(t,h)&:=\inf\{s\geq t\,:\,V_{k,\alpha}(s,r^{(\alpha)t,h}_s)=\Psi_k(s,r^{(\alpha)t,h}_s)\}\wedge T.
\end{align}
It can be shown that almost surely (possibly up to subsequences) $\tau^*_{k,\alpha,n,R}\wedge\tau^*_{k,\alpha,n}\to\tau^*_{k,\alpha,n}$ as $R\to\infty$, $\tau^*_{k,\alpha,n}\wedge\tau^*_{k,\alpha}\to\tau^*_{k,\alpha}$ as $n\to\infty$, and $\tau^*_{k,\alpha}\wedge\tau^*_{k}\to\tau^*_{k}$ as $\alpha\to\infty$ thanks to the uniform convergence in Theorems \ref{ascoli}, \ref{cnv2} and in Step 3 above. The proof follows along the same lines of the proof of Lemma \ref{taustar} below, hence we omit further details here (cf.~also \cite{Ch-DeA12a}, Lemma 4.7).

Now, set $\sigma:=\tau^*_{k,\alpha,n,R}\wedge\tau^*_{k,\alpha,n}\wedge\tau^*_{k,\alpha}\wedge\tau^*_{k}$ for simplicity, and recall that \eqref{sbruff} holds for $V^{(n)}_{k,\alpha,R}$, with $\tau^*_{k,\alpha,n,R}$ instead of $\tau^*_k$, by Step 1 above. Then a sufficient condition for the optimality of $\tau^*_{k}$ is
\begin{align}\label{suff}
\mathbb{E}\bigg\{e^{-\int^{\sigma}_t{r^{(\alpha)t,h;n}_u(0)du}}
V^{(n)}_{k,\alpha,R}\big(\sigma,r^{(\alpha)t,h;n}
_{\sigma}\big)\bigg\}\to\mathbb{E}\bigg\{e^{-\int^{\tau^*_{k}}_t{r^{t,h}_u(0)du}}
V_{k}\big(\tau^*_{k},r^{t,h}
_{\tau^*_{k}}\big)\bigg\}
\end{align}
as $R\to\infty$, $n\to\infty$ and $\alpha\to\infty$ in the prescribed order. Again, \eqref{suff} follows by the uniform convergence in Theorems \ref{ascoli}, \ref{cnv2} and in Step 3 above, and by the uniform integrability of Proposition \ref{unif-intk-na}. However here we skip further details as similar arguments will be used in the proof of Theorem \ref{optstt2} below.
\end{proof}

\section{A variational formulation for $V$}

Since $V_k$ converges to $V$ as $k\to\infty$ (cf.~Proposition \ref{regular}), it is natural to expect that $V$ might be a solution of a variational problem similar to \eqref{d-4}. However, that will be more likely to happen when taking limits of weaker variational problems, due to the lack of higher regularity of $\Psi$. In particular, notice that for all $v\in \mathcal{V}^p$, we may write $F_k(v)(t)$ (cf.~\eqref{ruiz0}) as
\begin{align}\label{ruiz}
F_k(v)(t)=(\frac{\partial\Psi_k}{\partial t}(t),v)_{\mu}-a_{\mu}(\Psi_k(t),v)=:\mathcal{A}_{k,\mu}(t;v),
\end{align}
by using arguments as in Remark \ref{buio}. 
Similarly, a continuous linear functional $\mathcal{A}_{\mu}(t;\,\cdot\,)\in \mathcal{V}^{p\,*}$, $t\in[0,T]$, associated to $\Psi$ may be defined by
\begin{align}
\mathcal{A}_{\mu}(t;v):=(\frac{\partial\Psi}{\partial t}(t),v)_{\mu}-a_{\mu}(\Psi(t),v),\qquad\forall v\in \mathcal{V}^p.
\end{align}
Then a simple application of Proposition \ref{mollify1} and Theorem \ref{uni-est} imply
\begin{equation}\label{prot2}
\lim_{k\to\infty}\Big\|\int_{t_1}^{t_2}{\big[\mathcal{A}_{k,\mu}(t;\,\cdot\,)-\mathcal{A}_{\mu}(t;\,\cdot\,)\big]dt}\Big\|
_{L^2(0,T;\mathcal{V}^{p})^*}=0\quad\text{for all $0\le t_1<t_2\le T$.}
\end{equation}
\begin{theorem}\label{vi-fin}
For every $1< p<\infty$ the function $\hat{u}:=V-\Psi$ is a solution of the variational problem \eqref{d-4} with $F_k(v-u(\,\cdot\,))(t)$ replaced by $\mathcal{A}_{\mu}(t;v-u(t))$.
\end{theorem}
\begin{proof}
The boundary conditions are clearly satisfied by $\hat{u}$, and the continuity of $\hat{u}$ is a consequence of Proposition \ref{lip0} and Corollary \ref{continuity-V}. It remains to prove that $\hat{u}$ solves the inequality in \eqref{d-4} when $F_k(v-u(\,\cdot\,))(t)$ is replaced by $\mathcal{A}_{\mu}(t;v-u(t))$. Recall that $\Psi_k$ and $V_k$ have the same bounds as $\Psi$ and $V$ in Proposition \ref{reg-psi2} and Corollary \ref{arto}, respectively. Then, for all $k\in\mathbb{N}$ and all $1< p<\infty$, $u_k$ of Theorem \ref{inf-yos-vi} and $\frac{\partial}{\partial t}u_k$ are bounded in $L^2(0,T;\mathcal{V}^p)$ and $L^2(0,T;L^p(\mathcal{H}_w,\mu))$, respectively, by a suitable constant $M_{\mu,p,T}>0$. Hence $u_k$ converges to some $\bar{u}$ as $k\to\infty$ weakly in $L^2(0,T;\mathcal{V}^p)$ and $\frac{\partial}{\partial t}u_k\rightharpoonup \frac{\partial}{\partial t}\bar{u}$ as $k\to\infty$ in $L^2(0,T;L^p(\mathcal{H}_w,\mu))$. However $u_k\to \hat{u}$ as $k\to\infty$ in $L^2(0,T;L^p(\mathcal{H}_w,\mu))$ for all $1\leq p<\infty$, by Propositions \ref{mollify1} and \ref{regular}. Therefore $\bar{u}=\hat{u}$.

Notice that, for all $v\in \mathcal{K}^p_{\mu}$, and arbitrary $0\le t_1<t_2\le T$ the function $u_k$ of Theorem \ref{inf-yos-vi} satisfies
\begin{align}\label{vecchio}
\hspace{-3pt}\displaystyle{\int_{t_1}^{t_2}\hspace{-4pt}\big[\hspace{-1pt}-\hspace{-1pt}(\frac{\partial u_k}{\partial t},v-u_k)_{\mu}+a_\mu(u_k,v-u_k)-\mathcal{A}_{k,\mu}(t;v-u_k)\big]dt\geq0}
\end{align}
by \eqref{d-4} and \eqref{ruiz}. In order to take limits as $k\to\infty$ in \eqref{vecchio} we write
\begin{align}\label{serra2}
\int_{t_1}^{t_2}&{a_\mu(u_k,u_k)dt}=\int_{t_1}^{t_2}{a_\mu(u_k-\hat{u},u_k-\hat{u})dt}+\int_{t_1}^{t_2}{a_\mu(\hat{u},u_k)dt}+
\int_{t_1}^{t_2}{a_\mu(u_k-\hat{u},\hat{u})dt}.
\end{align}
By using the bound $M_{\mu,p,T}$, \eqref{ext}, \eqref{bl-form} and estimates as in the proof of Theorem \ref{uni-est} we obtain
\begin{align}\label{serra1}
\int_{t_1}^{t_2}{\hspace{-7pt}a_\mu(u_k-\hat{u},u_k-\hat{u})dt}&\hspace{-1pt}\geq\hspace{-1pt}-C_{\mu,p}\big\|Du_k-D\hat{u}\big\|_{L^2(0,T;L^{2p'}
(\mathcal{H}_w,\mu;\mathcal{H}_w))}
\big\|u_k-\hat{u}\big\|_{L^2(0,T;L^{2p}(\mathcal{H}_w,\mu))}\nonumber\\
&\geq-C^\prime_{\mu,p,T}\,\big\|u_k-\hat{u}\big\|_{L^2(0,T;L^{2p}(\mathcal{H}_w,\mu))},
\end{align}
where $C^\prime_{\mu,p,T}=2\,C_{\mu,p}\,M_{\mu,p,T}$ and $\frac{1}{p}+\frac{1}{p'}=1$. Now we use \eqref{prot2}, \eqref{serra2}, \eqref{serra1} together with the convergence properties of $u_k$ and $\frac{\partial}{\partial t}u_k$ to obtain, in the limit,
\begin{align}
\hspace{-3pt}\displaystyle{\int_{t_1}^{t_2}\hspace{-4pt}\big[\hspace{-1pt}-\hspace{-1pt}(\frac{\partial \hat{u}}{\partial t},v-\hat{u})_{\mu}+a_\mu(\hat{u},v-\hat{u})-\mathcal{A}_{\mu}(t;v-\hat{u})\big]dt\geq0}.
\end{align}
Since $t_1$ and $t_2$ are arbitrary Theorem \ref{vi-fin} follows.
\end{proof}
Next we show that the stopping time
\begin{eqnarray}\label{stop}
\tau^*(t,h):=\inf\{s\geq t\,:\,V(s,r^{t,h}_s)=\Psi(s,r^{t,h}_s)\}\wedge T
\end{eqnarray}
is optimal for $V$ in \eqref{priceHJM2}. For that we need the next Lemma, whose proof follows along the lines of arguments adopted in \cite{Ben-Lio82}, Chapter 3, Section 3, Theorem 3.7 (cf.~in particular p. 322).
\begin{lemma}\label{taustar}
Let $\tau^*_{k}(t,h)$ be as in \eqref{rosp} and let $\tau^*(t,h)$ be as in \eqref{stop}. Then
\begin{eqnarray}\label{optstt}
\lim_{k\to\infty}\tau^*_{k}(t,h)\wedge\tau^*(t,h)=\tau^*(t,h),\quad\mathbb{P}\textrm{-a.s.}
\end{eqnarray}
\end{lemma}
\begin{proof}
For simplicity we consider the diffusion $r^h$ that starts at time zero from $h$. There is no loss of generality as all results below hold for arbitrary initial time $t$. We set $\tau^*_k:=\tau^*_{k}(0,h)$ and $\tau^*:=\tau^*(0,h)$. By Theorem \ref{inf-yos-vi}, $\tau^*_k$ is optimal for the $k$-th regularized problem. The limit \eqref{optstt} is trivial for those $\omega\in\Omega$ such that $\tau^*(\omega)=0$. Set $\Omega_0:=\big\{\omega\in\Omega\,:\,\tau^*(\omega)>0\big\}$. Fix $\omega\in\Omega_0$ and take $\delta<\tau^*(\omega)$. Then, for $t\in[0,\tau^*(\omega)-\delta]$,
\begin{eqnarray*}
V(t,r^{h}_t(\omega))>\Psi(t,r^{h}_t(\omega)).
\end{eqnarray*}
Since $t\mapsto r^{h}_t(\omega)$ is continuous, the continuous map $t\mapsto V(t, r^{h}_t(\omega))-\Psi(t, r^{h}_t(\omega))$ attains its minimum on $[0,\tau^*(\omega)-\delta]$; that is, there exists $\eta(\delta,\omega)>0$ such that
\begin{eqnarray*}
\eta(\delta,\omega):=\min\{V(t, r^{h}_t(\omega))-\Psi(t, r^{h}_t(\omega)),\,t\in[0,\tau^*(\omega)-\delta]\}
\end{eqnarray*}
and
\begin{eqnarray*}
V(t, r^{h}_t(\omega))\geq\Psi(t, r^{h}_t(\omega))+\eta(\delta,\omega),\qquad \textrm{for all $t\in[0,\tau^*(\omega)-\delta]$}.
\end{eqnarray*}
Recall that $\Psi_k\to\Psi$ and $V_k\to V$ uniformly on $[0,T]\times\mathcal{H}_w$ (cf.~Propositions \ref{mollify1} and \ref{regular}), therefore there exists $N_\eta=N(\eta(\delta,\omega))\in\mathbb{N}$ large enough and such that
\begin{eqnarray*}
V_k(t, r^{h}_t(\omega))>\Psi_k(t, r^{h}_t(\omega)),\qquad \textrm{for all $t\in[0,\tau^*(\omega)-\delta]$ and $k\geq N_\eta$.}
\end{eqnarray*}
It follows that $\tau^*(\omega)-\delta<\tau^*_{k}(\omega)$ for all $k\geq N_\eta$. Notice that $\eta(\delta,\omega)\to0$ as $\delta\to0$ and hence $N_\eta\to\infty$. Therefore
\begin{align*}
\lim_{\footnotesize
\begin{array}{c}
N_\eta\to\infty \\
k\geq N_\eta
\end{array}
}(\tau^*_{k}\wedge\tau^*)(\omega)=\tau^*(\omega).
\end{align*}
Since $\omega\in\Omega_0$ is arbitrary, \eqref{optstt} follows.
\end{proof}
\begin{theorem}\label{optstt2}
An optimal stopping time for $V$ in \eqref{priceHJM2} is given by $\tau^*(t,h)$ of \eqref{stop}.
\end{theorem}
\begin{proof}
Set $\tau^*=\tau^*(t,h)$ and $\tau^*_k=\tau^*_k(t,h)$ and take $\tau=\tau^*_k\wedge\tau^*$ in \eqref{sbruff} to obtain
\begin{eqnarray}\label{bufalo}
V_k(t,h)=\mathbb{E}\left\{e^{-\int_t^{\tau^*_k\wedge\tau^*}{r^{t,h}_u(0)du}}\,V_k(\tau^*_k\wedge\tau^*,r^{t,h}
_{\tau^*_k\wedge\tau^*})
\right\}.
\end{eqnarray}
By Proposition \ref{regular}, the left hand side of \eqref{bufalo} converges to $V(t,h)$ in the limit as $k\to\infty$. We now show the convergence of the right hand side. Take $(t,h)\in[0,T]\times\mathcal{H}_w$ and a sequence $(t_k,h_k)\to (t,h)$ as $k\to\infty$. Then
\begin{align*}
\big|V_k(t_k,h_k)-V(t,h)\big|\le& \big|V_k(t_k,h_k)-V_k(t_k,h)\big|\\
&+\big|V_k(t_k,h)-V(t_k,h)\big|+\big|V(t_k,h)-V(t,h)\big|\to0
\end{align*}
as $k\to\infty$ by Propositions \ref{regular} and \ref{rem-cont-Vk}, and the fact that $V_k$ has the same bounds as $V$ in Theorem \ref{reg-v-fun}. Hence,
\begin{align*}
e^{-\int_t^{\tau^*_k\wedge\tau^*}{r^{t,h}_u(0)du}}\,V_k(\tau^*_k\wedge\tau^*,r^{t,h}
_{\tau^*_k\wedge\tau^*})\to e^{-\int_t^{\tau^*}{r^{t,h}_u(0)du}}\,V(\tau^*,r^{t,h}
_{\tau^*})\qquad\mathbb{P}-\text{a.s.}
\end{align*}
as $k\to\infty$ by Lemma \ref{taustar}. Now, an application of Proposition \ref{unif-intk} provides
\begin{align*}
\lim_{k\to\infty}\mathbb{E}\left\{e^{-\int_t^{\tau^*_k\wedge\tau^*}{r^{t,h}_u(0)du}}\,V_k(\tau^*_k\wedge\tau^*,r^{t,h}
_{\tau^*_k\wedge\tau^*})\right\}=\mathbb{E}\left\{e^{-\int_t^{\tau^*}{r^{t,h}_u(0)du}}\,V(\tau^*,r^{t,h}
_{\tau^*})
\right\}
\end{align*}
and by taking limits in \eqref{bufalo} we conclude that
\begin{eqnarray*}
V(t,h)=\mathbb{E}\left\{e^{-\int_t^{\tau^*}
{r^{t,h}_u(0)du}}\,V(\tau^*,r^{t,h}_{\tau^*})\right\}=\mathbb{E}\left\{e^{-\int_t^{\tau^*}
{r^{t,h}_u(0)du}}\Psi(\tau^*,r^{t,h}_{\tau^*})\right\}.
\end{eqnarray*}
The optimality of $\tau^*$ now follows.
\end{proof}
\begin{remark}
The variational problem in Theorem \ref{vi-fin} is degenerate as the second order differential operator is not uniformly elliptic. Degenerate variational inequalities require non-classical arguments even at a finite-dimensional level. The question of uniqueness is very hard to address when such degeneracy occurs especially in the generality of our setting. In Section 5 of \cite{Ch-DeA12a} a class of problems where uniqueness may be obtained is discussed. There the key argument requires the operator $A$ to be self-adjoint, hence extending such result to the Musiela's parametrization of the HJM model is not straightforward and we leave it for future research.
\end{remark}

\begin{appendix}
\section*{Appendix}
\section{Proof of Theorem \ref{reg-v-fun}}\label{app-1}
\renewcommand{\theequation}{A-\arabic{equation}}
\begin{proof}
The first claim follows from \eqref{psi1} and Lemma \ref{bound-discount}. To show \eqref{lip-V} take $h,g\in\mathcal{H}_w$ and fix $t\in[0,T]$. Then \eqref{priceHJM3} implies
\begin{align}\label{lip-V1}
V(t,h)-V(t,g) \leq &\sup_{t\leq\tau\leq T}\mathbb{E}\left\{D(t,\tau;r^{t,h}_\cdot(0))\Psi(\tau,r^{t,h}_{\tau})-D(t,\tau;r^{t,g}_\cdot(0))
\Psi(\tau,r^{t,g}_{\tau})\right\}\nonumber\\
\le & K\sup_{t\leq\tau\leq T}\mathbb{E}\bigg\{\Big|D(t,\tau;r^{t,h}_\cdot(0))-D(t,\tau;r^{t,g}_\cdot(0))\Big|\bigg\}\\
&+ \sup_{t\leq\tau\leq T}\mathbb{E}\bigg\{D(t,\tau;r^{t,g}_\cdot(0))\Big|\Psi(\tau,r^{t,h}_{\tau})-\Psi(\tau,r^{t,g}_{\tau})\Big|\bigg\}.\nonumber
\end{align}
Now, recall \eqref{injection} and notice that
\begin{align}\label{lip-V2}
\Big|D(t,\tau;r^{t,h}_\cdot(0))&-D(t,\tau;r^{t,g}_\cdot(0))\Big|\nonumber\\
\le&\, \Big(D(t,\tau;r^{t,h}_\cdot(0))+D(t,\tau;r^{t,g}_\cdot(0))\Big)\Big|\int_t^\tau{\big(r^{t,h}_u(0)-r^{t,g}_u(0)\big)du}
\Big|\nonumber\\
\le&\: C\,T\Big(D(t,\tau;r^{t,h}_\cdot(0))+D(t,\tau;r^{t,g}_\cdot(0))\Big)\sup_{t\le s\le T}\big\|r^{t,h}_s-r^{t,g}_s\big\|_w.
\end{align}
Hence, Proposition \ref{lip0}, Lemma \ref{bound-discount}, \eqref{lip-V1}, \eqref{lip-V2} and H\"older's inequality give
\begin{align}\label{lip-V3}
V(t,h)-V(t,g)\le L_1 \,\gamma^{\frac{1}{2}}\Big(e^{\frac{\beta}{2}\|h\|_w}+e^{\frac{\beta}{2}\|g\|_w}\Big)\mathbb{E}\bigg\{\sup_{t\le s\le T}\big\|r^{t,h}_s-r^{t,g}_s\big\|^2_w\bigg\}^{\frac{1}{2}}
\end{align}
where $L_1:=C_1+C\,K\,T$. Since the same arguments hold for $V(t,g)-V(t,h)$, we obtain
\begin{equation*}
|V(t,h)-V(t,g)|\leq L_1\,\gamma^{\frac{1}{2}}\Big(e^{\frac{\beta}{2}\|h\|_w}+e^{\frac{\beta}{2}\|g\|_w}\Big)\mathbb{E}\bigg\{\sup_{t\leq s\leq T}\|r^{t,h}_s-r^{t,g}_s\|^2_{w}\bigg\}^{\frac{1}{2}}.
\end{equation*}
The coefficients in \eqref{infiniteSDEb} are time-homogeneous, hence (cf.~\eqref{lip})
\begin{align*}
\mathbb{E}\bigg\{\sup_{t\leq s\leq T}\|r^{t,h}_s-r^{t,g}_s\|^2_{w}\bigg\} &= \mathbb{E}\bigg\{\sup_{0\leq s\leq T-t}\|r^{h}_s-r^{g}_s\|^2_{w}\bigg\}\leq C_{2,T}\|h-g\|^2_w.
\end{align*}
Condition \eqref{lip-V} follows with $L_V=\, L_1\,\sqrt{\gamma\,C_{2,T}}$.

To prove \eqref{lip-Vt} we we need to write problem \eqref{priceHJM2} in a more convenient form. For any stopping time $\tau\in[t,T]$ with respect to the filtration $(\mathcal{F}^t_s)_{s\ge t}=\sigma\big\{B_s-B_t\,;\,s\ge t\big\}$ we set $\sigma:=\tau-t$. A change of variable gives
\begin{align*}
\mathbb{E}&\bigg\{e^{-\int_t^\tau{r^{t,h}_u(0)\,du}}\Big[K-e^{-\int^{T-\tau}_0{r^{t,h}_\tau(x)\,dx}}\Big]^+\bigg\}=
\mathbb{E}\bigg\{e^{-\int_0^{\sigma}{r^{t,h}_{t+u}(0)\,du}}\Big[K-e^{-\int^{T-t-\sigma}_0{r^{t,h}_{t+\sigma}(x)\,dx}}
\Big]^+\bigg\}
\end{align*}
and hence (cf.~\eqref{priceHJM2})
\begin{align}\label{timeLip03}
V(t,h)=\sup_{0\le \sigma\le T-t}\mathbb{E}\bigg\{e^{-\int_0^{\sigma}{r^{t,h}_{t+u}(0)\,du}}\Big[K-e^{-\int^{T-t-\sigma}_0{r^{t,h}_{t+\sigma}(x)\,dx}}
\Big]^+\bigg\}.
\end{align}
Note that the SDE \eqref{infiniteSDEb} has time-homogeneous coefficients and its solution is adapted to the filtration $(\mathcal{F}^t_s)_{s\ge t}$. Denoting $B^t_s:=B_s-B_t$, $s\ge t$, it follows that $r^{t,h}_{\,\cdot\,}=F(h, B^t)(\,\cdot\,)$ for a suitable Borel-measurable function $F$ (cf.~for instance \cite{Kall}, Lemma 21.15). Moreover $F(h, B^t)(\,\cdot\,)$ has the same law as $r^{h}_{\,\cdot\,}=F(h, B)(\,\cdot\,)$ (cf.~again \cite{Kall}, Lemma 21.15) and the set of the $(\mathcal{F}^t_s)_{s\ge t}$-stopping times is equivalent to the set of the $(\mathcal{F}_t)_{t\ge0}$-stopping times. Therefore \eqref{timeLip03} may be equivalently written as
\begin{align}\label{timeLip00}
V(t,h)=\sup_{0\le \sigma\le T-t}\mathbb{E}\bigg\{e^{-\int_0^{\sigma}{r^{h}_{u}(0)\,du}}\Big[K-e^{-\int^{T-t-\sigma}_0{r^{h}_{\sigma}(x)\,dx}}
\Big]^+\bigg\}
\end{align}
where the supremum is taken over the $(\mathcal{F}_t)_{t\ge0}$-stopping times.

Take $0\le t_1<t_2\le T$, then by arguing as in \eqref{lip-V1} and using \eqref{injection} and \eqref{lip1} we have
\begin{align}\label{timeLip01}
V(t_2,h)-V(t_1,h)\le& K\sup_{0\le\tau\le T-t_2}\mathbb{E}\bigg\{e^{-\int^\tau_0{r^h_u(0)du}}\Big|\int^{T-t_1-\tau}_{T-t_2-\tau}{r^h_\tau(x)dx}\Big|\bigg\}\nonumber\\
\le&K\,C\,(t_2-t_1)\,\mathbb{E}\bigg\{e^{-\int^\tau_0{r^h_u(0)du}}\sup_{0\le t\le T}\|r^h_t\|_w \bigg\}.
\end{align}
To find a bound for $V(t_1,h)-V(t_2,h)$ notice that $\tau\wedge(T-t_2)$ is admissible for $V(t_2,h)$ if $\tau\in[0,T-t_1]$. Then, again as for \eqref{lip-V1}, we get
\begin{align}
V&(t_1,h)-V(t_2,h)\\
\le& K\sup_{0\le\tau\le T-t_1}\mathbb{E}\bigg\{\Big(e^{-\int^\tau_0{r^h_u(0)du}}+e^{-\int^{\tau\wedge(T-t_2)}_0{r^h_u(0)du}}\Big)\Big|\int^{\tau}_{
\tau\wedge(T-t_2)}{\hspace{-5pt}r^h_\tau(x)dx}\Big|\bigg\}\nonumber\\
&+\sup_{0\le\tau\le T-t_1}\mathbb{E}\bigg\{e^{-\int^{\tau\wedge(T-t_2)}_0{r^h_u(0)du}}\Big|\int_0^{T-t_2-\tau\wedge(T-t_2)}
{r^h_{\tau\wedge(T-t_2)}(x)dx}-\int_0^{T-t_1-\tau}
{r^h_{\tau}(x)dx}\Big|
\bigg\}.\nonumber
\end{align}
The above expectations simplify by splitting their arguments over the events $\big\{\tau\le T-t_2\big\}$ and $\big\{\tau> T-t_2\big\}$; hence by recalling \eqref{injection} we obtain
\begin{align}\label{timeLip02}
V(t_1,h)-V(t_2,h)\le& K\,C(t_2-t_1)\sup_{0\le\tau\le T-t_1}\mathbb{E}\bigg\{\Big(e^{-\int^\tau_0{r^h_u(0)du}}+e^{-\int^{\tau\wedge(T-t_2)}_0{r^h_u(0)du}}\Big)\|r^h_\tau\|_w
\bigg\}\nonumber\\
&+2\,C(t_2-t_1)\sup_{0\le\tau\le T-t_1}\mathbb{E}\bigg\{\|r^h_\tau\|_w\,\,e^{-\int^{\tau\wedge(T-t_2)}_0{r^h_u(0)du}}
\bigg\}.
\end{align}
At this point in the right-hand side of \eqref{timeLip01} and \eqref{timeLip02} we use H\"older's inequality, \eqref{apr1}, \eqref{bound-r0} to obtain \eqref{lip-Vt} with a suitable $L'_V$.
\end{proof}

\section{Properties of the gain function}\label{app0}
\renewcommand{\theequation}{B-\arabic{equation}}
In order to prove Proposition \ref{reg-psi2} we need some auxiliary results about the regularity of $\Psi$. Let $\{\varphi_1,\varphi_2,\ldots\}$ be a set of orthonormal basis functions of ${\mathcal{H}_w}$ and for $n\geq1$ let $P_n:{\mathcal{H}_w}\to{\mathcal{H}_w}$ be the projection map defined by
\begin{equation}\label{projection}
P_nh:=\sum^n_{i=1}{\langle h,\varphi_i\rangle_w\,\varphi_i},
\end{equation}
where $\langle\cdot,\cdot\rangle_w$ is the scalar product in ${\mathcal{H}_w}$. Set $h^{(n)}:=P_nh$ and define $\Psi^{(n)}:[0,T]\times{\mathcal{H}_w}\to\mathbb{R}$ by
\begin{equation}\label{psi-n}
\Psi^{(n)}(t,h):=\Psi(t,P_nh)=\Psi(t,h^{(n)}).
\end{equation}
Dini's Theorem (cf.~\cite{Die}, Theorem 7.2.2) and \eqref{lip-psi1} give
\begin{equation}\label{psi-n2}
\lim_{n\to\infty}\sup_{(t,h)\in[0,T]\times\mathcal{K}}\big|\Psi^{(n)}(t,h)-\Psi(t,h)\big|=0,\quad\textrm{for every compact}\:
\mathcal{K}\subset{\mathcal{H}_w}.
\end{equation}
We will show now that $\Psi^{(n)}$ belongs to a suitable Sobolev space. The arguments of the proof are similar to those employed in \cite{DaPr}, Chapter 10.
\begin{lemma}\label{reg-psi}
For $1\le p <+\infty$ there exists $C_{\Psi,p}>0$ such that
\begin{equation}\label{rp1}
\sup_{t\in[0,T]}\|\Psi^{(n)}(t)\|_{W^{1,p}({\mathcal{H}_w},\mu)}<C_{\Psi,p}\qquad\text{and}\qquad
\int_0^T{\bigg\|\frac{\partial\Psi^{(n)}}{\partial t}(t)\bigg\|^2_{L^p({\mathcal{H}_w},\mu)}dt}\,<C_{\Psi,p}
\end{equation}
for all $n\in\mathbb{N}$.
\end{lemma}
\begin{proof}
Fix $1\le p<+\infty$. The uniform bound $K$ of (\ref{psi1}) gives
\begin{eqnarray*}
\sup_{t\in[0,T]}\|\Psi^{(n)}(t)\|_{L^p({\mathcal{H}_w},\mu)}\leq K\mu({\mathcal{H}_w})=K, \quad\textrm{for}\:\: n\in\mathbb{N}.
\end{eqnarray*}
Notice that $\Psi^{(n)}$ is a function defined on $[0,T]\times\mathbb{R}^n$; 
that is $\Psi^{(n)}(t,h)\equiv\Psi^{(n)}(t,h_1,\ldots h_n)$ for $h_i:=\langle h,\varphi_i\rangle_w$, $i=1,\ldots n$. Hence we may mollify $\Psi^{(n)}$ by the standard mollifiers $(\rho_k)_{k\in\mathbb{N}}$. In fact, fix $t\in[0,T]$ and define $\Psi^{(n)}_k(t,\cdot):=\rho_k\star \Psi^{(n)}(t,\cdot)$. Clearly the pointwise convergence holds, $\Psi^{(n)}_k(t,h)\to\Psi^{(n)}(t,h)$ as $k\to\infty$, for $h\in{\mathcal{H}_w}$ (cf. \cite{Brezis}, Proposition 4.21) and
\begin{eqnarray*}
|\Psi^{(n)}_k(t,z)|=\big|\int_{\mathbb{R}^n}{\rho_k(y)\Psi^{(n)}(t,z-y)dy}\big|\leq K\qquad\textrm{for $z:=(h_1,\ldots,h_n)\in\mathbb{R}^n$}
\end{eqnarray*}
by \eqref{psi1}. Hence 
$\Psi^{(n)}_k\to\Psi^{(n)}$ as $k\to\infty$ in $L^p(0,T;L^p({\mathcal{H}_w},\mu))$ by dominated convergence.

It is easy to see that the mollified functions $\Psi^{(n)}_k$ are equi-Lipschitz in the space variable, uniformly with respect to $t$, with the same constant $C_1$ (cf. \eqref{lip-psi1}). Therefore $\|D\Psi^{(n)}_k(t)\|_{L^p({\mathcal{H}_w},\mu;{\mathcal{H}_w})}=\|D\Psi^{(n)}_k(t)\|_{L^p(\mathbb{R}^n,\mu_n;\mathbb{R}^n)}\leq C_1\mu_n(\mathbb{R}^n)=C_1$, for all $n$, $k$ and $t$, by Remark \ref{conn}. We may conclude that
\begin{equation}\label{b01}
\|\Psi^{(n)}_k(t)\|_{L^p({\mathcal{H}_w},\mu)}+\|D\Psi^{(n)}_k(t)\|_{L^p({\mathcal{H}_w},\mu;{\mathcal{H}_w})}\leq (K+C_1)\qquad\textrm{for $n,k\in\mathbb{N}$ and $t\in[0,T]$.}
\end{equation}
Now \eqref{b01} guarantees that for each $n\in\mathbb{N}$ and $t\in[0,T]$ there exists a function $\phi^{(n)}(t)\in W^{1,p}({\mathcal{H}_w},\mu)$ and a subsequence $(\Psi^{(n)}_{k_j}(t))_{j\in\mathbb{N}}$ such that $\Psi^{(n)}_{k_j}(t)\rightharpoonup\phi^{(n)}(t)$ in $W^{1,p}({\mathcal{H}_w},\mu)$ as $j\to\infty$ (cf.~\cite{Brezis}, Theorem 3.18). It must be $\Psi^{(n)}(t)=\phi^{(n)}(t)$ by uniqueness of the limit, since $\Psi^{(n)}_k(t)\to\Psi^{(n)}(t)$ as $k\to\infty$ in $L^p({\mathcal{H}_w},\mu)$ and hence $\Psi^{(n)}(t)\in W^{1,p}({\mathcal{H}_w},\mu)$ for $t\in[0,T]$. The lower semi-continuity of the weak limit and (\ref{b01}) imply the first estimate of \eqref{rp1} with $C_{\Psi,p}=K+C_1$.

Similarly, to show the second bound in \eqref{rp1} we recall that $\int_{\mathcal{H}_w}{\|h\|^p_w\mu(d\,h)}<\infty$, apply the same arguments as before and use the Lipschitz property of $\Psi^{(n)}(\,\cdot\,,h)$, locally with respect to $h\in{\mathcal{H}_w}$ (cf.~\eqref{lip-psi2}). 
\end{proof}

\begin{proof}[Proof of \textbf{\textrm{Proposition \ref{reg-psi2}}}]
Take the projection $\Psi^{(n)}$ of $\Psi$ as in \eqref{psi-n}. The sequence $(\Psi^{(n)}(t,\cdot))_{n\in\mathbb{N}}$ is bounded in $W^{1,p}({\mathcal{H}_w},\mu)$, uniformly in $t\in[0,T]$, for all $1\le p<+\infty$ by Lemma \ref{reg-psi}. Therefore for every $t\in[0,T]$ there exists $\Phi(t)\in W^{1,p}({\mathcal{H}_w},\mu)$ such that $\Psi^{(n)}(t)\rightharpoonup\Phi(t)$ in $W^{1,p}({\mathcal{H}_w},\mu)$ as $n\to\infty$. By \eqref{psi-n2} and dominated convergence $\Psi^{(n)}(t)\to\Psi(t)$ in $L^p({\mathcal{H}_w},\mu)$, $1\leq p<\infty$ as $n\to\infty$ for all $t\in[0,T]$; hence $\Phi(t)=\Psi(t)$ for all $t\in[0,T]$. Moreover, since the first bound in \eqref{rp1} is uniform with respect to $n\in\mathbb{N}$ and $t\in[0,T]$, the lower semi-continuity of the weak convergence gives the first bound in \eqref{rp1-2}.

To prove the second bound in \eqref{rp1-2} it suffices to adopt arguments as above and use the second bound in \eqref{rp1}.
\end{proof}
\end{appendix}


\end{document}